\newif\ifpublic
\publicfalse 
\newif\ifIEEEtr
\IEEEtrfalse 

\ifIEEEtr
\documentclass[journal]{IEEEtran}
\else
\documentclass[10pt]{article}
\fi

\usepackage[pagebackref,colorlinks=true,linkcolor=blue,urlcolor=blue,citecolor=blue,pdfstartview=FitH]{hyperref}

\usepackage{amsmath}
\usepackage{amssymb}
\usepackage{amsthm,bbm}

\ifIEEEtr
\renewenvironment{proof}{\begin{IEEEproof}}{\end{IEEEproof}}

\else
\usepackage{fullpage}
\usepackage{mathpazo}
\fi

\usepackage[nameinlink,capitalize]{cleveref}
\renewcommand{\eqref}[1]{\hyperref[#1]{(\ref*{#1})}}

\newtheorem{theorem}{Theorem}[section]
\newtheorem{corollary}[theorem]{Corollary}
\newtheorem{lemma}[theorem]{Lemma}

\newtheorem{definition}[theorem]{Definition}
\newtheorem{claim}[theorem]{Claim}
\newtheorem{fact}[theorem]{Fact}
\theoremstyle{definition}
\newtheorem{remark}[theorem]{Remark}

\newcommand{\prob}[2]{\mathop{\mathrm{Pr}}_{#1}\left[#2\right]}
\newcommand{\avg}[2]{\mathop{\mathbb{E}}_{#1}\left[#2\right]}

\newcommand{\etal}{{\em et~al.}}
\newcommand{\Supp}{\mathrm{Supp}}

\newcommand{\LM}{\mathrm{LM}}
\newcommand{\RM}{\mc{P}}

\newcommand{\F}{\mathbb{F}}

\newcommand{\Test}{\mathrm{Test}}

\newcommand{\mc}[1]{\mathcal{#1}}

\ifpublic
\usepackage[disable]{todonotes}
\else
\usepackage[colorinlistoftodos]{todonotes}
\fi

\title{Robust Multiplication-based Tests for Reed-Muller Codes\thanks{A preliminary version of this paper appeared in the {\em Proc. 36th IARCS Conf. on Foundations of Software Technology \& Theoretical Computer Science (FSTTCS)}, 2016~\cite{HarshaS2016-rm}.}}

\ifIEEEtr
\author{Prahladh Harsha%
\thanks{TIFR, Mumbai,
    India. \texttt{prahladh@tifr.res.in}}
 and Srikanth Srinivasan\thanks{Department of Mathematics, IIT Bombay, Mumbai, India.
\texttt{srikanth@math.iitb.ac.in}}}
\else
\author{Prahladh Harsha\thanks{TIFR, Mumbai, India. \texttt{prahladh@tifr.res.in}} \and Srikanth Srinivasan\thanks{Department of Mathematics, IIT Bombay, Mumbai, India.
\texttt{srikanth@math.iitb.ac.in}}}
\date{}
\fi

\ifIEEEtr
\fi

\begin{document}

\maketitle

\begin{abstract}
  We consider the following multiplication-based tests to check if a
  given function $f: \F_q^n\to \F_q$ is a codeword of the Reed-Muller
  code of dimension $n$ and order $d$ over the finite field $\F_q$ for
  prime $q$ (i.e., $f$ is the evaluation of a degree-$d$ polynomial
  over $\F_q$ for $q$ prime).
\begin{itemize}
\item $\Test_{e,k}$: Pick $P_1,\ldots,P_k$ independent random degree-$e$
  polynomials and accept iff the function $fP_1\cdots P_k$ is the
  evaluation of a degree-$(d+ek)$ polynomial (i.e., is a codeword
  of the Reed-Muller code of dimension $n$ and order $(d+ek)$).
\end{itemize}
We prove the robust soundness of the above tests for large values of
$e$, answering a question of Dinur and Guruswami [Israel Journal of Mathematics, 209:611-–649, 2015]. Previous
soundness analyses of these tests were known only for the case when
either $e=1$ or $k=1$. Even for the case $k=1$ and $e>1$, earlier
soundness analyses were not robust.

We also analyze a derandomized version of this test, where (for
example) the polynomials $P_1,\ldots,P_k$ can be the \emph{same}
random polynomial $P$. This generalizes a result of Guruswami \etal\
[SIAM J. Comput., 46(1):132-–159, 2017].

One of the key ingredients that go into the proof of this robust
soundness is an extension of the standard Schwartz-Zippel lemma over
general finite fields $\F_q$, which may be of independent interest.
\end{abstract}

\section{Introduction}

\ifIEEEtr\else The \fi  problem of local testing of
codes~\cite{RubinfeldS1996,Arora1994,FriedlS1995,GoldreichS2006} has
received a lot of attention over the last two decades. Informally
speaking, the problem of testing a code is to design a robust
algorithmic procedure that tests if a given received word is a member
of the code or not. The algorithmic procedure on access to the
received word, queries it at a few random locations and decides to
accept/reject the word such that (1) all codewords are accepted by the
procedure and (2) non-codewords are rejected with probability
proportional to their distance from the code. The Reed-Muller code,
due to its inherent local characterization, is extremely amenable to
efficient local testing. Local testing of Reed-Muller codes over large
fields was extensively investigated in the
90's~\cite{RubinfeldS1996,FriedlS1995,AroraS2003,RazS1997}, primarily motivated
by their application to construction of probabilistically checkable
proofs~\cite{FeigeGLSS1996,BabaiFLS1991,AroraS1998,AroraLMSS1998}. More
recently, local testing of Reed-Muller codes over small fields have
also been
investigated~\cite{AlonKKLR2005,KaufmanR2006,BhattacharyyaKSSZ2010,HaramatySS2013}.

The basic problem of Reed-Muller code testing is to check if a given function $f: \F_q^n \to \F_q$ is
close to a degree-$d$ multivariate polynomial (over $\F_q$, the finite
field of $q$ elements)  or equivalently if the word $f$ is close to
the Reed-Muller code $\RM_q(n,d)$ of order $d$ and dimension $n$ over the finite
field $\F_q$. This problem, in its local testing version, for the case
when $q=2$
was first studied by Alon, Kaufman, Krivilevich, Litsyn and
Ron~\cite{AlonKKLR2005}, who proposed and analyzed a natural
$2^{d+1}$-query test for this problem. Subsequent to this work, improved analyses and
generalizations to larger fields were
discovered~\cite{KaufmanR2006,BhattacharyyaKSSZ2010,HaramatySS2013}. These tests
and their analyses led to several applications, especially in hardness
of approximation, which in turn spurred other Reed-Muller testing
results (which were not necessarily local tests)~\cite{DinurG2015,GuruswamiHHSV2017}. In this work, we give a robust
version of one of these latter multiplication based tests due to
Dinur and Guruswami~\cite{DinurG2015}. Below we describe this variation
of the testing problem, its context, and our results.

\subsection{Local Reed-Muller tests}

Given a field $\F_q$ of size $q$, let
$\mc{F}_q(n) := \{f\mid f:\F_q^n\rightarrow \F_q\}$. The Reed-Muller
code $\RM_q(n,d)$, parametrized by two parameters $n$ and $d$, is
the subset of $\mc{F}_q(n)$ that corresponds to those functions which
are evaluations of polynomials of degree at most $d$. If $n$, $d$ and
$q$ are clear from context, we let $r := (q-1)n-d$.

The proximity of two functions $f, g \in \mc{F}_q(n)$ is measured by
the Hamming distance. Specifically, we let $\Delta(f,g)$
denote the absolute  Hamming distance
between $f$ and $g$, i.e., $\Delta(f,g) := \#\{x \in \F_q^n \mid f(x)
\neq g(x)\}$. For a family of functions $\mc{G} \subseteq \mc{F}_q(n)$, we
let $\Delta(f,\mc{G}) : = \min\{ \Delta(f,g)\mid g \in \mc{G}\}$. We
say that $f$ is $\Delta$-close to $\mc{G}$ if $\Delta(f,\mc{G}) \leq
\Delta$ and $\Delta$-far otherwise.

The following natural local test to check membership of a function $f$ in
$\RM_2(n,d)$ was proposed by Alon~\etal~\cite{AlonKKLR2005} for the
case when $q=2$ (and extended by Kaufman and Ron~\cite{KaufmanR2006} to larger $q$).
\begin{itemize}
\item AKKLR Test: Input $f:\F_2^n \to \F_2$
\begin{itemize}
\item Pick a random $d+1$-dimensional affine space $A$. 
\item Accept iff $f|_{A} \in \RM_2(d+1,d)$.
\end{itemize}
\end{itemize}

Here, $f|_{A}$ refers to the restriction of the function $f$ to the
affine space $A$. Bhattacharyya~\etal~\cite{BhattacharyyaKSSZ2010}
showed the following optimal analysis of this test.
\begin{theorem}[\cite{AlonKKLR2005,BhattacharyyaKSSZ2010}]\label{thm:bkssz}
There exists an absolute constant $\alpha >0$ such that the following
holds. If
$f\in\mc{F}_2(n)$ is $\Delta$-far from $\RM_2(n,d)$ for $\Delta\in
\mathbb{N}$, then 
$$\Pr_A[f|_A \not\in \RM_2(d+1,d)] \geq \min\{\Delta/2^r, \alpha\} .$$
\end{theorem}

Subsequent to this result, Haramaty, Shpilka and
Sudan~\cite{HaramatySS2013} extended this result to all constant sized
fields $\F_q$. These optimal analyses then led to the discovery of the
so-called ``short code'' (aka the low degree long code) due to
Barak~\etal~\cite{BarakGHMRS2015} which has played an important role
in several improved hardness of approximation
results~\cite{DinurG2015,GuruswamiHHSV2017,KhotS2017,Varma2015,Huang2015}.

\subsection{Multiplication-based tests}

We now consider the following type of multiplication-based tests to
check membership in $\RM_q(n,d)$, parametrized by two numbers $e, k
\in \mathbb{N}$. 

\begin{itemize}
\item $\Test_{e,k}$: Input $f: \F_q^n \to \F_q$
\begin{itemize}
\item Pick $P_1,\ldots,P_k\in_R \RM_q(n,e)$. 
\item Accept iff $fP_1\cdots P_k\in \RM_q(n,d+ek)$.
\end{itemize}
\end{itemize}

This tests computes the point-wise product of $f$ with $k$ random
degree-$e$ polynomials $P_1,\ldots,P_k$ respectively and checks that
the resulting product function $fP_1\cdots P_k$ is the evaluation of a
degree-$(d+ek)$ polynomial. Unlike the previous test, this test is not
necessarily a local test.

The key lemma due to Bhattacharyya~\etal~\cite{BhattacharyyaKSSZ2010}
that led to the optimal analysis in \cref{thm:bkssz} is the
following robust analysis of $\Test_{1,1}$.

\begin{lemma}[\cite{BhattacharyyaKSSZ2010}]\label{lem:bkssz}
Let $f\in\mc{F}_2(n)$ be $\Delta$-far from $\RM_2(n,d)$ for $\Delta
= 2^r/100$. For randomly picked $\ell \in \RM_2(n,1)$, we have
\[
\prob{\ell}{\Delta(f\cdot \ell,\RM_2(n,d+1)) < \beta \Delta} =
O\left(\frac{1}{2^{r}}\right) ,
\]
for some absolute constant $\beta > 0$.
\end{lemma}

Observe that the AKKLR test is equivalent to $\Test_{1,r-1}$ for
$r=n-d$. This observation coupled with a simple inductive argument
using the above lemma implies \cref{thm:bkssz}.

Motivated by questions related to hardness of coloring hypergraphs,
Dinur and Guruswami studied the $\Test_{e,1}$ for $e = r/4$ and proved
the following result.

\begin{lemma}[\cite{DinurG2015}]\label{lem:dg}
Let $f\in\mc{F}_2(n)$ be $\Delta$-far from $\RM_2(n,d)$ for $\Delta
= 2^r/100$ and let $e = (n-d)/4$. For randomly picked $P \in \RM_2(n,e)$, we have
\[
\prob{P}{f\cdot P \in \RM_2(n,d+e)} \leq  \frac{1}{2^{2^{\Omega(e)}}}.
\]
\end{lemma}

Note that the $\Test_{e,1}$ is not a local test (as is the case with multiplication
based tests of the form $\Test_{e,k}$). Furthermore, the above lemma does not
give a robust analysis unlike \cref{lem:bkssz}. More precisely, the
lemma only bounds the probability that the product function
$f\cdot P$ is in $\RM_2(n,d+e)$, but does not say anything about
the probability of $f\cdot P$ being close to $\RM_2(n,d+e)$ as in
\cref{lem:bkssz}. Despite this, this lemma has had several
applications, especially towards proving improved inapproximability
results for hypergraph
colouring~\cite{DinurG2015,GuruswamiHHSV2017,KhotS2017,Varma2015,Huang2015}.

\subsection{Our results}

Our work is motivated by the question raised at the end of the
previous section: can the analysis of the Dinur-Guruswami Lemma be
strengthened to yield a robust version of \cref{lem:dg}? Such a
robust version, besides being interesting of its own right, would
yield a soundness analysis of the $\Test_{e,k}$ for $k > 1$ (wherein the input
function $f$ is multiplied by $k$ degree-$e$ polynomials). This is
similar to how \cref{lem:bkssz} was instrumental in proving
\cref{thm:bkssz}.

We begin by first showing this latter result (ie., the soundness
analysis of the $\Test_{e,k}$). First for some notation. For
non-negative $n$ and $d$, let $N_q(n,d)$ denote the number of monomials $m$ in indeterminates
$X_1,\ldots,X_n$ such that the degree of each variable in $m$ is at
most $q-1$ and the total degree is at most $d$. Equivalently,
$N_q(n,d)$ is the dimension of the vector space $\RM_q(n,d).$ For $n <
0$, we define $N_q(n,d) = 1$.

\begin{theorem}[Soundness of $\Test_{e,k}$]
\label{thm:test-e-k-new}
For every prime $q$ there exists a constant $c_q$ such that the
following holds. Let $k\in\mathbb{N}$ be arbitrary
constant. Let
$n,d,r,\Delta,e\in\mathbb{N}$ be positive integers such that $r=(q-1)n-d$,
$\Delta \leq q^{r/4(q-1)-2}$, and
$e\leq r/4k$. Then, given any $f\in \mc{F}_q(n)$ that
is $\Delta$-far from $\RM_q(n,d)$ and for $P_1,\ldots,P_k$ chosen
independently and uniformly at random from $\RM_q(n,e)$, we have
\[
\prob{P_1,\ldots,P_k}{fP_1P_2\cdots P_k \in \RM_q(n,d+ek)} \leq \frac{k}{q^{N}},
\]
for some $N\geq \eta(q,k)\cdot N_q(\lfloor\frac{L}{10}\rfloor-c_q,e),$ for $L = \lfloor \log_q \Delta\rfloor$ and $\eta(q,k) = \frac{1}{q^{k/q-1}\ln q}.$ 
\end{theorem}

\begin{remark}
\begin{itemize}
\item To appreciate the parameters of \cref{thm:test-e-k-new}, it is
  instructive to lower bound the probability
  $\prob{P_1,\ldots,P_k}{fP_1P_2\cdots P_k \in \RM_q(n,d+ek)}$ for some fixed $f$. Let $L$ and $\Delta$ be positive
  integers such that $\Delta=q^L$. Let $f:\F_q^n \to \F$ be the
  function defined  as follows:
$$f(x_1,\dots,x_n) = \prod_{i=1}^{n-L} \left( 1-x_i^{q-1}\right).$$
By definition, $f\in \RM_q(n,D)$ where $D:= (n-L)(q-1)$. Hence, it has
distance $\Delta=q^L$ from any other degree-$D$ polynomial (see
\cref{fac:ring}) and hence,
also from all functions in $\RM_q(n,d)$ provided $D \geq d$. Observe
that $f$ is a function which is one on the $L$-dimensional subspace $V:=\{x_1
= x_2 = \cdots = x_{n-L} =0\}$ and 0 elsewhere. Let us consider what
happens when we multiply $f$ with $P_1P_2\cdots P_k$ where each of the
$P_i$ are random functions in $\RM_q(n,e)$. Let $E_i$ denote the event
that $P_i$ vanishes on the subspace $V$. We first note that $\Pr[E_i]
= q^{-N_q(L,e)}$. We then observe that if any of the events
$E_i$ happen, then $fP_1P_2\cdots P_k \equiv 0$. Hence, the
$\prob{P_1,\ldots,P_k}{fP_1P_2\cdots P_k \in \RM_q(n,d+ek)}$ is lower
bounded by the probability $\prob{}{\exists i, \; E_i} \approx
kq^{-N_q(L,e)}$. \cref{thm:test-e-k-new} states that this is roughly the largest that it can be.
\item The constant $c_q$ is obtained from a result of Haramaty,
  Shpilka and Sudan~\cite{HaramatySS2013} (see statement of \cref{lem:HSS}). 
\item The quantity $N_q(n,d)$ is the number of distinct monomials in
  $n$ variables of individual degree at most $(q-1)$ and total degree
  at most $d$. This is certainly lower bounded by the number of
  monomials in $\lfloor d/(q-1) \rfloor$ variables of individual
  degree at most $q-1$, which is exactly $q^{\lfloor d/(q-1)
    \rfloor}$. Plugging this bound of $N_q$ into
  \cref{thm:test-e-k-new} yields the following corollary.
\end{itemize}
\end{remark}

\begin{corollary}
\label{thm:test-e-k}
Let $q,k\in\mathbb{N}$ be constants with $q$ prime and $\varepsilon,\delta\in (0,1)$ be arbitrary constants. Let
$n,d,r,\Delta,e\in\mathbb{N}$ be such that $r=(q-1)n-d$,
$q^{\varepsilon r}\leq \Delta \leq q^{r/4(q-1)-2}$, and
$\delta r \leq e\leq r/4k$. Then, given any $f\in \mc{F}_q(n)$ that
is $\Delta$-far from $\RM_q(n,d)$ and for $P_1,\ldots,P_k$ chosen
independently and uniformly at random from $\RM_q(n,e)$, we have
\[
\prob{P_1,\ldots,P_k}{fP_1P_2\cdots P_k \in \RM_q(n,d+ek)} \leq \frac1{q^{q^{\Omega(r)}}},
\]
where the $\Omega(\cdot)$ above hides a constant depending on $k,q,\delta,\varepsilon$.
\end{corollary}

We then show that the above corollary can be used to prove the
following robust version of \cref{lem:dg}, answering an open question
of Dinur and Guruswami~\cite{DinurG2015}.

\begin{theorem}[Robust soundness of $\Test_{e,1}$]\label{lem:robustdg}
Let $q\in\mathbb{N}$ be a constant with $q$ prime and $\varepsilon,\delta\in (0,1)$ be arbitrary constants. Let $n,d,r,\Delta, e\in\mathbb{N}$ be such that $r=(q-1)n-d$,
$q^{\varepsilon r}\leq \Delta \leq q^{r/4(q-1)-2}$, and
$\delta r \leq e\leq r/8$. Then, there is a $\Delta' = q^{\Omega(r)}$ such that given any $f\in \mc{F}_q(n)$ that
is $\Delta$-far from $\RM_q(n,d)$ and for $P$ chosen uniformly at random from $\RM_q(n,e)$, we have
\[
\prob{P}{\Delta(f\cdot P ,\RM_q(n,d+e)) < \Delta'} \leq \frac{1}{q^{q^{\Omega(r)}}}\enspace,
\]
where the $\Omega(\cdot)$ above hide constants depending on $q,\delta,\varepsilon$.
\end{theorem}

Equipped with such multiplication-based tests, we can ask if one can
prove the soundness analysis of other related multiplication-based
tests. For instance, consider the following test which checks
correlation of the function $f$ with the square of a random degree-$e$
polynomial. 

\begin{itemize}
\item $\text{Corr-Square}_{e}$: Input $f: \F_3^n \to \F_3$
\begin{itemize}
\item Pick $P\in_R \RM_3(n,e)$. 
\item Accept iff $f \cdot P^2\in \RM_3(n,d+2e)$.
\end{itemize}
\end{itemize}

This test was used by Guruswami~\etal~\cite{GuruswamiHHSV2017} to
prove the hardness of approximately coloring 3-colorable 3-uniform
hypergraphs. However, their analysis was restricted to the squares of
random polynomials. Our next result shows that this can be extended to
any low-degree polynomial of  random polynomials. More precisely, let
$h \in \RM_q(1,k)$ be a \emph{univariate} polynomial of degree exactly $k$ for some $k <
q$. Consider the following test.

\begin{itemize}
\item $\text{Corr-$h$}_{e}$: Input $f: \F_q^n \to \F_q$
\begin{itemize}
\item Pick $P\in_R \RM_q(n,e)$. 
\item Accept iff $f \cdot h(P)\in \RM_q(n,d+ek)$.
\end{itemize}
\end{itemize}

We show that an easy consequence of \cref{thm:test-e-k} proves
the following soundness claim about the test $\text{Corr-$h$}$.
\begin{corollary}[Soundness of $\text{Corr-$h$}_e$]\label{cor:corrq}
Let $q,k\in\mathbb{N}$ be constants with $q$ prime, $k < q$, and let $\varepsilon,\delta\in (0,1)$ be arbitrary constants.\footnote{The assumption $k<q$ is necessary here since otherwise $h(P)$ could be $P^q-P$, which is always $0$.} Let
$n,d,r,\Delta,e\in\mathbb{N}$ be such that $r=(q-1)n-d$, 
$q^{\varepsilon r}\leq \Delta \leq q^{r/4(q-1)-2}$, and
$\delta r \leq e\leq r/4k$. Let $h \in \RM_q(1,k)$ be a univariate
polynomial of degree
exactly $k$. Then, given any $f\in \mc{F}_q(n)$ that
is $\Delta$-far from $\RM_q(n,d)$ and for $P$ chosen uniformly at random from $\RM_q(n,e)$, we have
\[
\prob{P }{f\cdot h(P) \in \RM_q(n,d+ek)} \leq \frac1{q^{q^{\Omega(r)}/2^k}},
\]
where the $\Omega(\cdot)$ above hides a constant depending on $k,q,\delta,\varepsilon$.
\end{corollary}

\paragraph*{A generalization of the Schwartz-Zippel lemma over $\F_q$.} A special case of \cref{thm:test-e-k-new} is already quite interesting. This case corresponds to when the function $f$ is a polynomial of degree exactly $d'$, for some $d'$ slightly larger than $d$. (It is quite easy to see by the Schwartz-Zippel lemma over $\F_q$ --- which guarantees that a non-zero polynomial of low degree is non-zero at many points --- that this $f$ is far from $\RM_q(n,d)$.) In this case, we would expect that when we multiply $f$ with $k$ random polynomials $P_1,\ldots,P_k\in \RM_q(n,e)$, that the product $fP_1\cdots P_k$ is a polynomial of degree exactly $d'+ek$ and hence not in $\RM_q(n,d+ek)$ with high probability. 

We are able to prove a tight version of this statement (\cref{lem:fmaxmon}). For every degree $d'$, we find a polynomial $f$ of degree exactly $d'$ that maximizes the probability that $fP_1$ has degree $< d'+s$ for any parameter $s\leq e$. This polynomial turns out to be the same polynomial for which the Schwartz-Zippel lemma over $\F_q$ is tight. This is not a coincidence: it turns out that our lemma is a generalization of the Schwartz-Zippel lemma over $\F_q$ (see \cref{sec:SZconn}).

Given  the utility of the Schwartz-Zippel lemma in Coding theory and Theoretical Computer Science, we think this statement may be of independent interest.

\subsection{Proof ideas}

The basic outline of the proof of \cref{thm:test-e-k-new} is similar to the proof of \cref{lem:dg} from the work of Dinur and Guruswami~\cite{DinurG2015} which corresponds to \cref{thm:test-e-k-new} in the case that $q=2$ and $k=1$.  We describe this argument in some detail so that we can highlight the variations in our work.

The argument is essentially an induction on the parameters $e, r=n-d,$ and $ \Delta$. As long as $r$ is a sufficiently large constant, \cref{lem:bkssz} can be used~\cite[Lemma 22]{DinurG2015} to show that for any $f\in \mc{F}_2(n)$ that is $\Delta$-far from $\RM_2(n,d)$, there is a variable $X$ such that for each $\alpha \in \{0,1\} = \F_2$, the restricted function $f|_{X=\alpha}$ is $\Delta' = \Omega(\Delta)$-far from $\RM_2(n-1,d)$.\footnote{Actually, \cref{lem:bkssz} implies the existence of a linear function with this property and not a variable. But after a linear transformation of the underlying space, we may assume that it is a variable.}

Now, to argue by induction, we write 
\begin{equation}
\label{eq:fnPdecomp}
f = Xg + h \text{ and } P_1 = XQ_1 + R_1
\end{equation}
where $g,h,Q_1,R_1$  depend on $n-1$ variables, $Q_1$ is a random polynomial of degree $\leq e-1$ and $R_1$ is a random polynomial of degree $\leq e$. Using the fact that $X^2 = X$ over $\F_2$, we get $fP_1= X((g+h)Q_1 + gR_1) + hR_1$. 

Since $f|_{X=\alpha}$ is $\Delta'$-far from $\RM_2(n-1,d)$, we see
that both $h$ and $g+h$ are $\Delta'$-far from $\RM_2(n-1,d)$. To
apply induction, we note that $fP_1\in \RM_2(n,d+e)$ iff $hR_1\in
\RM_2(n-1,d+e)$ and $(g+h)Q_1 +
hR_1\in \RM_2(n-1,d+e-1)$; we call these events $\mc{E}_1$ and $\mc{E}_2$ respectively. We bound the
overall probability by $\prob{}{\mc{E}_1}\cdot \prob{}{\mc{E}_2\mid R_1}$ (note that $\mc{E}_1$ depends only on $R_1$).

We first observe that $\prob{}{\mc{E}_1}$ can be immediately bounded
using the induction hypothesis since $h$ is $\Delta'$-far from
$\RM_q(n-1,d+e)$ and $R_1$ is uniform over $\RM_q(n-1,e)$. The
second term $\prob{}{\mc{E}_2\mid R_1}$ can also be bounded by the
induction hypothesis with the following additional argument. We argue that (for
any fixed $R_1$) the probability that
$(g+h)Q_1 + gR_1\in \RM_2(n-1,d+e-1)$ is bounded by the probability
that $(g+h)Q_1\in \RM_2(n-1,d+e-1)$: this follows from the fact
that the number of solutions to any system of linear equations is
bounded by the number of solutions of the corresponding homogeneous
system (obtained by setting the constant term in each equation to
$0$). Hence, it suffices to bound the probability that
$(g+h)Q_1\in \RM_2(n-1,d+e-1)$, which can be bounded by the
induction hypothesis since $(g+h)$ is $\Delta'$-far from
$\RM_2(n-1,d)$ and $Q_1$ is uniform over $\RM_2(n-1,e-1)$ and we
are done.

Though our proofs follow the above template, we need to deviate from
the proof above in some important ways which we elaborate below.

The first is the decomposition of $f$ and $P_1$ from
\eqref{eq:fnPdecomp} obtained above, which yields two events
$\mc{E}_1$ and $\mc{E}_2$, the first of which depends only on $R_1$
and the second on both $Q_1$ and $R_1$. For $q > 2$, the standard
monomial decomposition of polynomials does not yield such a nice
``upper triangular'' sequence of events. So we work with a different
polynomial basis to achieve this. This choice of basis is closely
related to the polynomials for which the Schwartz-Zippel lemma over
$\F_q$ is tight. While such a basis was used in the special case of
$q=3$ in the work of Guruswami \etal\ ~\cite{GuruswamiHHSV2017}
(co-authored by the authors of this work), it was done in a somewhat
ad-hoc way. Here, we give, what is in our opinion a more transparent
construction that additionally works for all $q$.

Further modifications to the Dinur-Guruswami argument are required to handle $k>1$. We illustrate this with the example of $q=2$ and $k=2$. Decomposing as in the Dinur-Guruswami argument above, we obtain $f = Xg + h$, $P_1= XQ_1 + R_1$, and $P_2 = XQ_2 + R_2$. Multiplying out, we get
\ifIEEEtr
\else
\[
fP_1P_2 = X(\underbrace{Q_1Q_2(g+h) + (g+h)(Q_1R_2 + Q_2R_1) + gR_1R_2}_{ Q:=}) + hR_1R_2\enspace.
\]
\fi

Bounding the probability that $fP_1P_2\in \RM_2(n,d+2e)$ thus reduces to bounding the probability of event that $hR_1R_2\in \RM_2(n-1,d+2e)$ --- $\mc{E}_1$ depending only on $R_1$ and $R_2$ --- and then the probability that $Q\in \RM_2(n-1,d+2e-1)$ --- denoted $\mc{E}_2$ --- given any fixed $R_1$ and $R_2$. The former probability can be bounded using the induction hypothesis straightforwardly.

By a reasoning similar to the $k=1$ case, we can reduce bounding $\prob{}{\mc{E}_2\mid R_1,R_2}$ to the probability that $Q_1Q_2(g+h)\in \RM_2(n-1,d+2e-1)$. However, now we face a problem. Note that we have $g+h=f|_{X = 1}$ is $\Delta'$-far from $\RM_2(n-1,d)$ and $Q_1,Q_2\in \RM_2(n-1,e-1)$. Thus, the induction hypothesis only allows us to upper bound the probability that $Q_1Q_2(g+h)\in \RM_2(n-1,d+2e-2)$ which is not quite the event that we want to analyze. Indeed, if $f$ is a polynomial of degree exactly $d+1$, then the polynomial $Q_1Q_2(g+h)\in \RM_2(n-1,d+2e-1)$ with probability $1$. A similar problem occurs even if $f$ is a polynomial of degree $d'$ slightly larger than $d$ or more generally, when $f$ is \emph{close} to some polynomial of degree $d'$.

This naturally forces us to break the analysis into two cases. In the first case, we assume not just that $f$ is far from $\RM_2(n,d)$ but also from $\RM_2(n,d')$ but for some $d'$ a suitable parameter larger than $d$. In this case, we can modify the proof of Dinur and Guruswami to bound the probability that $fP_1P_2\in \RM_2(n,d+2e)$ as claimed in \cref{thm:test-e-k-new}. In the complementary case when $f$ is close to some polynomial $F\in \RM_2(n,d')$, we can essentially assume that $f$ \emph{is} a polynomial of degree exactly $d'$. In this case, we can use the extension of Schwartz-Zippel lemma referred to above to show that with high probability $fP_1P_2$ is in fact a polynomial of degree exactly $d'+2e$ and is hence not of degree $d+2e < d'+2e$.

\subsection{Organization}

We begin with some notation and definitions in \cref{sec:prelim}. We
prove the extension of the Schwartz-Zippel lemma (\cref{lem:fmaxmon})
in \cref{sec:SZext} and then \cref{thm:test-e-k-new} in
\cref{sec:testek}. Finally, we give two applications of \cref{thm:test-e-k} in \cref{sec:applns}: one to proving a robust version of the above test (thus resolving a question of Dinur and Guruswami~\cite{DinurG2015}) and the other to proving \cref{cor:corrq}.

\section{Preliminaries}
\label{sec:prelim}

For a prime power $q$, let $\F_q$ denote the finite field of size
$q$. We use $\F_q[X_1,\ldots,X_n]$ to denote the standard polynomial
ring over variables $X_1,\ldots,X_n$ and $\RM_q(n)$ to denote the
ring $\F_q[X_1,\ldots,X_n]/\langle X_1^q-X_1,\ldots,X_n^q-X_n\rangle$.

We can think of the elements of $\RM_q(n)$ as elements of
$\F_q[X_1,\ldots,X_n]$ of individual degree at most $q-1$ in a natural
way. Given $P,Q\in \RM_q(n)$, we use $P\cdot Q$ or $PQ$ to denote
their product in $\RM_q(n)$. We use $P*Q$ to denote their product
in $\F_q[X_1,\ldots, X_n]$.

Given a set $S\subseteq \F_q^n$ and an $f\in \RM_q(n)$, we use $f|_S$ to denote the restricted function on the set $S$. Typically, $S$ will be specified by a polynomial equation. One special case is the case when $S$ is a hyperplane: i.e., there is a non-zero homogeneous degree-$1$ polynomial $\ell(X)\in \RM_q(n)$ and an $\alpha\in \F_q$ such that $S = \{x\mid \ell(x) = \alpha\}$. In this case, it is natural to think of $f|_{\ell(X) = \alpha} = f|_S$ as an element of $\RM_q(n-1)$ by applying a linear transformation that transforms $\ell(X)$ into the variable $X_n$  and then setting $X_n = \alpha$.

For $d\geq 0$, we use $\RM_q(n,d)$ to denote the polynomials in $\RM_q(n)$ of degree at most $d$. 

The following are standard facts about the ring $\RM_q(n)$ and the space of functions mapping $\F_q^n$ to $\F_q$.

\begin{fact}
\label{fac:ring}
\begin{enumerate}
\item Consider the ring of functions mapping $\F_q^n$ to $\F_q$ with addition and multiplication defined pointwise. This ring is isomorphic to $\RM_q(n)$ under the natural isomorphism that maps each polynomial $P\in \RM_q(n)$ to the function (mapping $\F_q^n$ to $\F_q$) represented by this polynomial.
\item In particular, each function $f:\F_q^n\rightarrow\F_q$ can be represented uniquely as a polynomial from $\RM_q(n)$. As a further special case, any non-zero polynomial from $\RM_q(n)$ represents a non-zero function $f:\F_q^n\rightarrow\F_q$.
\item (Schwartz-Zippel lemma over $\F_q$~\cite{KasamiLP1968}) Any non-zero polynomial from $\RM_q(n,d)$ is non-zero on at least $q^{n-a-1}(q-b)$ points from $\F_q^n$ where $d = a(q-1)+b$  and $0\leq b < q-1$. 
\item In particular, if $f,g\in \RM_q(n,d)$ differ from each other at at most $\Delta < q^{n-a-1}(q-b)$ points, then $f=g$.
\item (A probabilistic version of the Schwartz-Zippel lemma~(see, e.g., \cite{HaramatySS2013})) It follows from the above that given a non-zero polynomial $g\in \RM_q(n,d)$, then $g(x)\neq 0$ at a uniformly random point of $\F_q^n$ with probability at least $q^{-d/(q-1)}$. Similarly, if $f,g\in \RM_q(n,d)$ are distinct, then for uniformly random $x\in \F_q^n$, the probability  that $f(x) \neq g(x)$ is at least $q^{-d/(q-1)}$.
\end{enumerate}
\end{fact}

From now on, we will use without additional comment the fact that functions from $\F_q^n$ to $\F_q$ have unique representations as multivariate polynomials where the individual degrees are bounded by $q-1$.

Recall that $m_1*m_2$ denotes the product of
these monomials in the ring $\F_q[X_1,\ldots,X_n]$ while $m_1\cdot m_2$ denotes their product in $\RM_q(n) = \F_q[X_1,\ldots,X_n]/\langle X_1^q-X_1,\ldots,X_n^q-X_n \rangle$. We say that monomials $m_1,m_2\in \RM_q(n)$ are \emph{disjoint} if
$m_1*m_2 = m_1\cdot m_2$ (where the latter monomial is interpreted naturally as an element of $\F_q[X_1,\ldots,X_n]$). Equivalently, for each variable $X_i$ ($i\in [n]$), the sum of its degrees in $m_1$ and $m_2$ is less than $q$.

Given distinct monomials $m_1,m_2\in \F_q[X_1,\ldots,X_n]$, we say
that $m_1> m_2$ if either one of the following holds:
$\deg(m_1)> \deg(m_2)$, or $\deg(m_1) = \deg(m_2)$ and we have
$m_1 = \prod_i X_i^{e_i}$ and $m_2 = \prod_i X_i^{e_i'}$ where for the
least $j$ such that $e_j \neq e_j'$, we have $e_j > e_j'$.

The above is called the \emph{graded lexicographic} order on monomials~\cite{CoxLittleOShea}. This ordering obviously restricts to an ordering on the monomials in $\RM_q(n)$, which are naturally identified as a subset of the monomials of $\F_q[X_1,\ldots,X_n]$. The well-known fact about this monomial ordering we will use is the following.
\begin{fact}[\cite{CoxLittleOShea}]
\label{fac:mon-order}
For any monomials $m_1,m_2,m_3$, we have $m_1 \leq m_2 \Rightarrow m_1*m_3 \leq m_2*m_3$.
\end{fact}

Given an $f\in \RM_q(n)$, we use $\Supp(f)$ to denote the set of points $x\in \F_q^n$ such that $f(x)\neq 0$. If $f\neq 0$, we use $\LM(f)$ to denote the largest monomial (w.r.t. ordering defined above) with non-zero coefficient in $f$.

Let $m = \prod_{i\in [n]} X_i^{e_i}$ with $e_i < q$ for each $i$ and let $d = \deg(m)$. For an integer $s \geq 0$, we let
\begin{align*}
U_{s}(m) &:= \{\prod_{j\in [n]}X_j^{e_j'} \mid \sum_j e_j' = d+s \text{ and } \forall j\ q > e_j'\geq e_j,\},\\
D_s(m) &:= \{\prod_{j\in [n]}X_j^{e_j'} \mid  \sum_j e_j' = s \text{ and } \forall j\  e_j' + e_j < q\}.
\end{align*}

Note that the monomials in $D_s(m)$ are precisely the monomials of degree $s$ that are disjoint from $m$. Further, the map $\rho: D_s(m) \rightarrow U_s(m)$ defined by $\rho(m_1) = m_1\cdot m$ defines a bijection between $D_s(m)$ and $U_s(m)$, and hence we have

\begin{fact}
\label{fac:U=D}
For any monomial $m$ and any $s\geq 0$, $|U_s(m)| = |D_s(m)|$.
\end{fact} 

For non-negative integers $s\leq e$, we define
$U_{s,e}(m) := \bigcup_{s\leq t\leq e}U_t(m)$ and
$D_{s,e}(m) := \bigcup_{s\leq t\leq e}D_{t}(m)$. Since
$|U_t(m)| = |D_t(m)|$ for each $t$, we have
$|U_{s,e}(m)| = |D_{s,e}(m)|$.

\subsection{A different basis for $\RM_q(n)$}
\label{sec:basis}

Applying \cref{fac:ring} in the case that $n=1$, it follows that the monomials $\{X^i\mid 0\leq i < q\}$ form a natural basis for the space of all functions from $\F_q$ to $\F_q$. The following is another such basis which is sometimes more suitable for our purposes. 

\begin{definition}[A suitable basis for the space of functions from $\F_q$ to $\F_q$]
\label{def:basis}
Fix a linear ordering $\preceq$ of all the elements of $\F_q$. Let $\xi_0,\ldots,\xi_{q-1}$ be the elements of $\F_q$ according to this ordering. For any $i\in \{0,\ldots,q-1\}$, let $b_i^{\preceq}(X) = \prod_{j < i}(X-\xi_j)$. Note that for $i < q$, $b_i^{\preceq}(X)$ is a non-zero polynomial of degree $i$. In particular, $\{b_i^{\preceq}(X)\mid 0 \leq i < q\}$ is a basis for the space of all functions from $\F_q$ to $\F_q$. Usually, when we apply this definition, the ordering $\preceq$ will be implicitly clear and hence we will use $b_i(X)$ to refer to $b_i^{\preceq}(X)$.
\end{definition}

The following property of this basis will be useful.
\begin{lemma}
\label{lem:basis-property}
Fix any ordering $\preceq$ of $\F_q$ and let $\{b_i(X)\mid 0\leq i < q\}$ be the corresponding basis as in \cref{def:basis}. Then, for any $f:\F_q\rightarrow \F_q$ and $i\in \{0,\ldots,q-1\}$, we have $f(X)\cdot b_i(X) = f(\xi_i)b_i(X) + b_i'(X)$ where $b_i'(x) \in \mathrm{span}\{b_{i+1}(X),\ldots, b_{q-1}(X)\}$.
\end{lemma}

\begin{proof}
We know that $f(X)$ is a polynomial of degree at most $q-1$ in $X$. By linearity, it suffices to prove the lemma for $f(X) = X^k$ for $0\leq k \leq q-1$. We prove this by induction on $k$. The base case ($k=0$) of the induction is trivial. We also handle the case $k=1$ by noting that
\[
X\cdot b_i(X) = \xi_ib_i(X) + (X-\xi_i)b_i(X) = \xi_i b_i(X) + b_{i+1}(X)
\]
which has the required form.

Now consider $k\in \{2,\ldots,q-1\}$. By the induction hypothesis, we
know that $X^{k-1}\cdot b_i(X) = \xi_i^{k-1}b_i(X) + b_i'(X)$ where
$b_i'(X)\in \mathrm{span}\{b_{i+1}(X),\ldots, b_{q-1}(X)\}$. Hence, we
see that $X^k\cdot b_i(X) = X\cdot \xi_i^{k-1}b_i(X) + Xb_i'(X) =
(X-\xi_i + \xi_i) \cdot \xi_i^{k-1}b_i(X) + Xb_i'(X)$. Expanding we
obtain
\ifIEEEtr
\else
\[
X^k\cdot b_i(X) = \xi_i^k b_i(X) + (X-\xi_i)b_i(X) + Xb_i'(X) = \xi_i^k b_i(X) + b_{i+1}(X) + Xb_i'(X) = \xi_i^k b_i(X) + b_i''(X)
\]
\fi
where $b_i''(X)\in \mathrm{span}\{b_{i+1}(X),\ldots, b_{q-1}(X)\}$ by using the fact that $Xb_i'(X)\in \mathrm{span}\{b_{i+1}(X),\ldots, b_{q-1}(X)\} $, which follows from the case $k=1$. This proves the induction statement and hence also the lemma.
\end{proof}

We now consider functions $f:\F_q^n\rightarrow\F_q$ over $n$ variables $X_1,\ldots,X_n$. As noted above, this space of functions is ring isomorphic to $\RM_q(n)$. We will use an alternate basis for this space also.

We fix an ordering $\preceq$ of $\F_q$ and let $\{b_{i}(X_j)\mid 0\leq  i< q\}$ be the corresponding basis in the variable $X_j$. We refer to functions of the form $\prod_{j\in [n]}b_{i_j}(X_j)$ as \emph{generalized monomials} w.r.t. $\preceq$: we call this set $\mc{B}_q(n)$ (the orderings will be implicit). The \emph{degree} of the monomial $\prod_{j\in [n]}b_{i_j}(X_j)$ is $\sum_{j\in [n]}i_j$. Given a degree parameter $d\in\mathbb{N}$, we let $\mc{B}_q(n,d)$  denote the set of all monomials in $\mc{B}_q(n)$ of degree at most $d$. 

The following fact is easily proved.

\begin{fact}
\label{fac:basis}
\begin{enumerate}
\item For any $n,d\in \mathbb{N}$, the set $\mc{B}_q(n,d)$ is a basis for the space of polynomials in $\RM_q(n,d)$.
\item In particular, the set $\mc{B}_q(n) = \mc{B}_q(n,(q-1)n)$ is a basis for $\RM_q(n)$.
\end{enumerate}
\end{fact}

What makes the above basis useful is the following lemma.
\begin{lemma}
\label{lem:ut-decomp}
Fix any ordering $\xi_0,\ldots,\xi_{q-1}$ of $\F_q$ and let $b_{i}(X)$ ($0\leq i \leq q-1$) be the corresponding basis. Given
 any $f\in\RM_q(n)$ and any $P\in \RM_q(n,d)$, we may write the function $f\cdot P\in \RM_q(n)$ as
\[
fP = \sum_{k=0}^{q-1}b_k(X_n) \left(Q_k\cdot f|_{X_n=\xi_k} + \sum_{0\leq j<k}Q_j\cdot h_{j,k}\right)
\]
where $P = \sum_{k=0}^{q-1}b_k(X_n)Q_k(X_1,\ldots,X_{n-1})$, and $h_{j,k}(X_1,\ldots,X_{n-1})\in \RM_q(n-1)$.
\end{lemma}

\begin{remark}
\label{rem:ut}
The above statement encapsulates the advantage of working with the basis from \cref{def:basis}. Note that the coefficient of $b_k(X_n)$ only involves $Q_i(X_1,\ldots,X_{n-1})$ for $i\leq k$. This gives us an ``upper triangular'' decomposition of the polynomial $fP$ that we will find useful.
\end{remark}

\begin{proof}
By \cref{fac:basis} point 1, we can write $f = \sum_{i=0}^{q-1}b_i(X_n) f_i(X_1,\ldots,X_{n-1})$. Expanding $fP$, we get

\begin{align*}
fP &= \sum_{i,j \in \{0,\ldots,q-1\}} b_i(X_n) b_j(X_n) f_iQ_j\\
\ifIEEEtr\else(\text{by \cref{lem:basis-property}})\fi&= \sum_{i,j}f_i Q_j\cdot \left(b_i(\xi_j)b_j(X_n) + \sum_{k >
  j}\alpha_{i,j,k} b_k(X_n) \right)\\
\\
&= 	\sum_{k=0}^{q-1}b_k(X_n) \left(Q_k \sum_{i}f_ib_i(\xi_k) + \sum_{j < k,i}\alpha_{i,j,k}f_iQ_j\right)\\
&= \sum_{k=0}^{q-1}b_k(X_n)\left(Q_k f|_{X_n=\xi_k} + \sum_{j<k}Q_j\cdot h_{j,k}\right),\\
\end{align*}
where $h_{j,k} := \sum_{i}\alpha_{i,j,k}f_i$.
\end{proof}

We will also need to analyze the product of many polynomials in the above basis, for which we use the following.

\begin{lemma}
\label{lem:prod-basis}
Say $P_1,\ldots, P_k\in \RM_q(n,d)$ with $P_i = \sum_{j = 0}^{q-1}b_j(X_n) Q_{i,j}(X_1,\ldots,X_{n-1})$. Let $P = \prod_{i=1}^k P_i = \sum_{j=0}^{q-1}b_j(X_n) Q_{j}(X_1,\ldots,X_{n-1})$. Given $j_1,\ldots,j_k \in \{0,\ldots,q-1\}$, we say that $(j_1,\ldots,j_k)\leq j$ if $j_i \leq j$ for each $i\in [k]$ and $(j_1,\ldots,j_k) < j$ if $j_i\leq j$ for each $i\in [k]$ \emph{and} there is some $i$ such that $j_i < j$. Also, let $Q_{(j_1,\ldots,j_k)}$ denote $\prod_{i\in [k]}Q_{i,j_i}$.

For each $j\in \{0,\ldots,q-1\}$, we have
\[
Q_j = \sum_{(j_1,\ldots,j_k)\leq j}\beta^{(j)}_{(j_1,\ldots,j_k)} Q_{(j_1,\ldots,j_k)},
\]
where $\beta^{(j)}_{(j_1,\ldots,j_k)} \in \F_q$ and further $\beta^{(j)}_{(j,\ldots,j)}\neq 0$.
\end{lemma}

\begin{proof}
We prove the lemma by induction on $k$. The base case $k = 1$ is trivial since we can take $\beta^{(j)}_{(j_1)} = 1$ if $j_1 = j$ and $0$ otherwise. 

Now, consider the inductive case $k > 1$. For $\tilde{P} = \prod_{i < k}P_i$, we have the above claim, which yields
\[
\tilde{Q}_j = \sum_{(j_1,\ldots,j_{k-1})\leq j}\tilde{\beta}^{(j)}_{(j_1,\ldots,j_{k-1})} Q_{(j_1,\ldots,j_{k-1})},
\]
where $\tilde{P} = \sum_j b_j(X_n) \tilde{Q}_j$. Also, $\tilde{\beta}^{(j)}_{(j,j,\ldots,j)}\neq 0$.

To prove the inductive claim, we expand $P = \prod_i P_i = \tilde{P} P_k$ and use \cref{lem:basis-property}. The computation is as follows.
\begin{align}
P = \tilde{P}P_k &= \left(\sum_{j}b_j(X_n)\tilde{Q}_j \right)\cdot \left(\sum_{\ell=0}^{q-1}b_\ell(X_n) Q_\ell\right)\notag\\
&= \sum_{j,\ell} \tilde{Q}_jQ_\ell b_j(X_n) b_\ell(X_n).\label{eq:prod-1}
\end{align}

By \cref{lem:basis-property}, it follows that
\[
b_j(X_n) b_\ell(X_n) = \sum_{r \geq (j,\ell)} \gamma^{(r)}_{(j,\ell)} b_r(X_n),
\]
where $\gamma^{(r)}_{(j,\ell)}\in \F_q$ for each $(j,\ell)\leq r$ and
in particular $\gamma^{(r)}_{(r,r)} = b_r(\xi_r)\neq 0$. Substituting
in \eqref{eq:prod-1} we get 
\begin{align*}
P &= \sum_{j,\ell} \tilde{Q}_j Q_\ell\sum_{r\geq (j,\ell)} \gamma^{(r)}_{(j,\ell)} b_r(X_n)\\
&= \sum_{r} b_r(X_n) \sum_{(j,\ell)\leq r}\gamma^{(r)}_{(j,\ell)} \tilde{Q}_j Q_\ell\\
\ifIEEEtr\else\text{(by Induction Hypothesis) }\fi&= \sum_{r} b_r(X_n)
                                                    \sum_{(j,\ell)\leq
                                                    r}\gamma^{(r)}_{(j,\ell)}Q_{\ell}\sum_{(j_1,\ldots,j_{k-1})\leq
                                                    j}
                                                    \tilde{\beta}^{(j)}_{\bar{j}}Q_{\bar{j}}\\
&= \sum_r b_r(X_n) \sum_{(j_1,\ldots,j_{k-1},\ell)\leq r} \beta^{(r)}_{(j_1,\ldots,j_{k-1},\ell)}Q_{(j_1,\ldots,j_{k-1},\ell)},
\end{align*}
where 
$$\beta^{(r)}_{(j_1,\ldots,j_{k-1},\ell)} = \sum_{j\geq (j_1,\ldots,j_{k-1}), j\leq r}\gamma^{(r)}_{(j,\ell)}\tilde{\beta}^{(j)}_{(j_1,\ldots,j_{k-1})}.$$ 
In particular, $\beta^{(r)}_{(r,\ldots,r)} = \gamma^{(r)}_{(r,r)}\tilde{\beta}^{(r)}_{(r,\ldots,r)}\neq 0$ since we showed that $\gamma^{(r)}_{(r,r)}\neq 0$ above and $\tilde{\beta}^{(r)}_{(r,\ldots,r)}\neq 0$ by the Induction Hypothesis.
\end{proof}

\subsection{Multilinear and set-multilinear systems of equations}
\label{sec:multilin}

Fix any set $\mc{Z}$ of variables and say we have a partition  $\Pi= \{\mc{Z}_1,\ldots, \mc{Z}_k\}$ of $\mc{Z}$. A polynomial $P\in \F_q[\mc{Z}]$ is \emph{$\Pi$-set-multilinear} (or just \emph{set-multilinear} if $\Pi$ is clear from context) if every monomial appearing in $P$ involves exactly one variable from each $\mc{Z}_i$ ($i\in [k]$). The polynomial $P$ is $\Pi$-multilinear if every monomial involves \emph{at most} one variable from each $\mc{Z}_i$ ($i\in [k]$). Note that a $\Pi$-set-multilinear polynomial is homogeneous of degree $k$ and a $\Pi$-multilinear polynomial has degree at most $k$.

Given a $\Pi$ as above and a $\Pi$-multilinear polynomial $P$, its homogeneous degree $k$ component is a $\Pi$-set-multilinear polynomial $Q$. We call $Q$ the \emph{set-multilinear part} of $P$.

\begin{lemma}
\label{lem:multilineqns}
Fix any set $\mc{Z} = \{Z_1,\ldots,Z_N\}$ of variables and a partition $\Pi = \{\mc{Z}_1,\ldots,\mc{Z}_k\}$ of $\mc{Z}$. Let $P_1,\ldots,P_m$ be any set of $\Pi$-multilinear polynomials with set-multilinear parts $Q_1,\ldots,Q_m$ respectively. Then, we have
\ifIEEEtr
\else
\[
\prob{z\sim \F_q^N}{P_1(z) = 0\wedge \cdots \wedge P_m(z) = 0} \leq \prob{z\sim \F_q^N}{Q_1(z) = 0\wedge \cdots \wedge Q_m(z) = 0}.
\]
\fi
\end{lemma}

The above lemma generalizes the well-known fact that a system of (inhomogeneous) linear equations has at most as many solutions as the corresponding \emph{homogeneous} system of linear equations obtained by setting the constant term in each equation to $0$. 

\begin{proof}
The proof uses the above fact about the number of solutions for systems of linear equations. Consider the following systems of multilinear polynomial equations. For $j \in \{0,\ldots,k\}$ and $i\in [m]$, define $P_{j,i}$  as follows: $P_{0,i} = P_i$ and given $P_{j,i}$ for $j < k$, we define $P_{j+1,i}$ by dropping all monomials from $P_{j,i}$ that do \emph{not} involve the variables from $\mc{Z}_{j+1}$. In particular, we see that $P_{k,i} = Q_i$ for each $i\in [m]$.

We claim that for each $j < k$ we have
\begin{equation}
\label{eq:jth}
\prob{z\sim \F_2^N}{\bigwedge_{i\in [m]}P_{j,i}(z) = 0} \leq \prob{z\sim \F_2^N}{\bigwedge_{i\in [m]}P_{j+1,i}(z) = 0}.
\end{equation}
The above clearly implies the lemma.

To show that \eqref{eq:jth} holds, we argue as follows. Fix any assignment to all the variables in $\mc{Z}\setminus\mc{Z}_{j+1}$. For each such assignment, the event on the Left Hand Side of \eqref{eq:jth} is the event that a system of $m$ linear equations $\mc{L}$ in $\mc{Z}_{j+1}$ is satisfied by a uniformly random assignment to $\mc{Z}_{j+1}$: this follows since each $P_{j,i}$ is a multilinear polynomial w.r.t. $\Pi$. On the Right Hand Side, we have the event that some other system $\mc{L}'$ of $m$ linear equations is satisfied. By inspection, it can be verified that $\mc{L}'$ is the homogeneous version of $\mc{L}$: i.e., each equation in $\mc{L}'$ is obtained by zeroing the constant term of the corresponding equation in $\mc{L}$. By standard linear algebra, $\mc{L}'$ has at least as many solutions as $\mc{L}$. Hence, the probability that a random assignment to the variables in $\mc{Z}_{j+1}$ satisfies $\mc{L}'$ is at least the probability that a random assignment satisfies $\mc{L}$. This implies \eqref{eq:jth}.
\end{proof}

\subsection{A result of Haramaty, Shpilka, and Sudan}

The following is an easy corollary of a result from the work of Haramaty, Shpilka, and Sudan~\cite{HaramatySS2013}. Analogous corollaries have been observed before by Dinur and Guruswami~\cite{DinurG2015} (using~\cite{BhattacharyyaKSSZ2010}) and Guruswami \etal~\cite{GuruswamiHHSV2017}.

\begin{lemma}
\label{lem:HSS}
Let $q$ be any constant prime. There is a constant $c_q > q$ depending only on $q$ such that the following holds. Let $n,d,\Delta,r$ be non-negative integers with $d < (q-1)n$, $r:= (q-1)n-d$, $q^{5}< \Delta < q^{r/(q-1)}$, and $r \geq c_q$. Then, for any $f\in \RM_q(n)$ that is $\Delta$-far from $\RM_q(n,d)$, there is a non-zero homogeneous linear function $\ell(X_1,\ldots,X_n)$ such that for each $\alpha\in \F_q$, the restriction $f|_{\ell(X) = \alpha}$ is at least $\Delta/q^3$-far from $\RM_q(n-1,d)$.
\end{lemma}

We need the following theorem due to Haramaty, Shpilka and Sudan~\cite{HaramatySS2013}.

\begin{theorem}[{\cite[Theorem~1.7 and 4.16]{HaramatySS2013}}\label{thm:HSS}
  using absolute distances instead of
  fractional distances] For every prime $q$, there exists a constant $\lambda_q$ such that
  the following holds.
For  $\beta: \F_q^n \to \F_q$, let $A_1,\dots, A_K$ be hyperplanes
such that $\beta|_{A_i}$ is $\Delta_1$-close to some degree $d$ polynomial
on $A_i$. If $K > q^{\lceil \frac{d+1}{q-1}\rceil+\lambda_q}$ and $\Delta_1 < q^{n-d/(q-1)-2}/2$,
then $\Delta(\beta,\RM_q(n,d)) \leq 2q\Delta_1 + 4(q-1)\cdot q^n/K$.
\end{theorem}

\begin{proof}[Proof of \cref{lem:HSS}]
Let $c_q = cq\lambda_q$ where $\lambda_q$ is the constant from \cref{thm:HSS} and $c$ is an absolute constant determined below.

Suppose \cref{lem:HSS} were false with $r \geq c_q$.
Then, for every nonzero homogeneous linear function $\ell$, at least one of
$\{f|_{\ell=\alpha }\mid \alpha \in \F_q\}$ is $\Delta/q^3$-close to a degree $d$ polynomial. We thus, get
$K=(q^n-1)/(q-1)$ hyperplanes such that the restriction of $f$ to these
hyperplanes is $\Delta/q^3$-close to a degree $d$
polynomial. Observe that $K \geq q^{n-1} > q^{\lceil \frac{d+1}{q-1}\rceil +\lambda_q}$ if $r \geq c_q$ and the constant $c$ is chosen large enough. Also note that since $\Delta < q^{r/(q-1)}$, we have $\Delta/q^3 < q^{(r/(q-1))-3} \leq q^{n-d/(q-1)-2}/2$. Hence, by \cref{thm:HSS} we
have  $\Delta(f, \RM_q(n,d)) \leq 2
\Delta/q^2 + 4\cdot(q-1)^2\cdot q^n/(q^{n}-1) < 2\Delta/q^2 + 8(q-1)^2 <
\Delta$ (since $\Delta \geq q^5$). This contradicts the hypothesis
that $f$ is $\Delta$-far from $\RM_q(n,d)$.
\end{proof}

\section{An extension of the Schwartz-Zippel Lemma over $\mathbb{F}_q$}
\label{sec:SZext}

The results of this section hold over $\F_q$ where $q$ is any prime power.

\begin{lemma}
\label{lem:Umax}
Let $d,s\geq 0$ be arbitrary integers with $d+s\leq n(q-1)$. Assume $d
= (q-1)u + v$ for $u,v \geq 0$ with $v < (q-1)$. Then the monomial
$m_0:=X_1^{q-1}\cdots X_u^{q-1}X_{u+1}^v$ of degree $d$ satisfies
$|U_s(m_0)| \leq |U_s(m)|$ for all monomials $m$ of degree exactly $d$.
\end{lemma}

\begin{proof}
Fix any monomial $m$ of degree $d$ such that $|U_s(m)|$ is as small as possible; say $m = \prod_{j\in [n]} X_j^{e_j}$. By renaming the variables if necessary, we assume that $e_1\geq e_2\geq\cdots \geq e_n$. 

If $m\neq m_0$, then we can find an $i< n$ such that $0< e_{i+1}\leq e_i < q-1$. Consider the monomial $m' = X_i^{e_{i}+1}X_{i+1}^{e_{i+1}-1}\prod_{j\not\in\{i,i+1\}}X_j^{e_j}$. We claim that $|U_s(m')|\leq |U_s(m)|$. This will complete the proof of the lemma, since it is easy to check that by repeatedly modifying the monomial in this way at most $d$ times, we end up with the monomial $m_0$. By construction, we will have shown that $|U_s(m_0)|\leq |U_s(m)|$.

We are left to show that $|U_s(m')|\leq |U_s(m)|$ or equivalently (by \cref{fac:U=D}) that
$|D_s(m')|\leq |D_s(m)|$. To this end, we show that for any
$(n-2)$-tuple
$\mathbf{e}' = (e_1',\ldots,e_{i-1}',e_{i+2}',\ldots,e_n')$, we have
$|D_s(m',\mathbf{e}')|\leq |D_s(m,\mathbf{e}')|$ where
$D_s(m,\mathbf{e}')$ denotes the set of monomials $\tilde{m}\in D_s(m)$ such that for
each $j\in [n]\setminus\{i,i+1\}$, the degree of $X_j$ in $\tilde{m}$ is $e_j'$. To
see this, note that $D_s(m,\mathbf{e}')$ and $D_s(m',\mathbf{e}')$ are
in bijective correspondence with the sets $S$ and $T$ respectively,
defined as follows:
\begin{align*}
S &= \{(d_1,d_2)\mid 0\leq d_1 \leq a, 0\leq d_2 \leq b, d_1+d_2 = c \},\\
T &= \{(d_1,d_2)\mid 0\leq d_1 \leq a-1, 0\leq d_2 \leq b+1, d_1+d_2 =
    c \},
\end{align*}
where $a := (q-1)-e_i$, $b := (q-1)-e_{i+1}$, and $c = s-\sum_{j\not\in\{i,i+1\}} e_j'$; note that by assumption, $(q-1) > e_i\geq e_{i+1}$ and hence $1\leq a\leq b$. Our claim thus reduces to showing $|T|\leq |S|$, which is done as follows.

If $c < 0$ or $c > a+b$, then both $S$ and $T$ are empty sets and the claim is trivial. So assume that $0\leq c\leq a+b$. In this case, we see that $|T\setminus S|\leq 1$: in fact, $T\setminus S$ can only contain the element $(c-b-1,b+1)$ and this happens only when the inequalities $0\leq c - b - 1\leq a-1$ are satisfied. But this allows us to infer that $S\setminus T$ contains $(a,c-a)$ since $0\leq c-b-1\leq c-a$ and $c-a \leq b$. Thus, $|T\setminus S|\leq |S\setminus T|$ and hence $|T|\leq |S|$.
\end{proof}

We have the following immediate corollary of \cref{lem:Umax}.
 
\begin{corollary}
\label{cor:Umax}
Let $d,e,s\geq 0$ be arbitrary parameters with $s\leq e$ and $d \leq
n(q-1)$. Assume $d = (q-1)u + v$ for $u,v \geq 0$ with $v <
(q-1)$. Then the monomial $m_0:=X_1^{q-1}\cdots X_u^{q-1}X_{u+1}^v$
satisfies $|U_{s,e}(m_0)|\leq |U_{s,e}(m)|$ for all monomials $m$ of
degree exactly $d$.
\end{corollary}

The main technical lemma of this section is the following. 

\begin{lemma}[Extension of the Schwartz-Zippel lemma over $\F_q$]
\label{lem:fmaxmon}
Let $e,d,s\geq 0$ be integer parameters with $s\leq e$. Let $f\in \RM_q(n)$ be non-zero and of degree \emph{exactly} $d$ with $\LM(f) = m_1$. Then,  
\[
\prob{P\in_R \RM_q(n,e)}{\deg(fP) < d+s}\leq \frac{1}{q^{|U_{s,e}(m_1)|}}.
\]
In particular, using \cref{cor:Umax}, the probability above is upper bounded by $\frac{1}{q^{|U_{s,e}(m_0)|}}$ where the monomial $m_0$ is as defined in the statement of \cref{cor:Umax}.
\end{lemma}

\begin{proof}
  Let $P = \sum_{m:\deg(m)\leq e}\alpha_m m$ where $m$ ranges over all monomials in $\RM_q(n)$ of degree at most $e$ and the $\alpha_m$ are
  chosen independently and uniformly at random from $\F_q$. Also, let
  $f = \sum_{i=1}^N \beta_i m_i$ where $\beta_i\neq 0$ for each $i$
  and we have $m_1 > m_2 > \cdots > m_N$ in the graded lexicographic
  order defined earlier.

Thus, we have 
\ifIEEEtr
\else
\begin{align*}
  fP &= \left(\sum_{m: \deg(m)\leq e} \alpha_m m\right) \cdot
       \left(\sum_{i=1}^N \beta_i m_i\right)
= \sum_{\tilde{m}} \left(\sum_{(m,j): mm_j = \tilde{m}} \alpha_m \beta_j\right) \tilde{m}.
\end{align*}
\fi
The polynomial $fP$ has degree $< d+s$ iff for each $\tilde{m}$ of
degree at least $d+s$, its coefficient in the above expression is
$0$. Since the $\beta_i$'s are fixed, we can view this event as the
probability that some set of \emph{homogeneous} linear equations in
the $\alpha_m$ variables (one equation for each $\tilde{m}$ of degree at least $d+s$) are satisfied. By standard linear algebra,
this is exactly $q^{-t}$ where $t$ is the rank of the linear
system. So it suffices to show that there are at least
$|U_{s,e}(m_1)|$ many \emph{independent} linear equations in the
system.

Recall that $|D_{s,e}(m_1)| = |U_{s,e}(m_1)|$. Now, for each
$m\in D_{s,e}(m_1)$, consider the ``corresponding'' monomial
$\tilde{m} = m\cdot m_1 = m * m_1\in U_{s,e}(m_1)$ (the second equality is true since
$m$ is disjoint from $m_1$). Note that each $\tilde{m}\in U_{s,e}(m_1)$ has
degree exactly $\deg(m) + \deg(m_1) \in [d+s,d+e]$. Thus, for $fP$ to
have degree $< d+s$, the coefficient of each $\tilde{m}$ must
vanish. Further, since
$|D_{s,e}(m_1)| = |U_{s,e}(m_1)|$ it suffices to
show that the linear equations corresponding to the different
$\tilde{m}\in U_{s,e}(m_1)$ are all linearly independent.

To prove this, we argue as follows. Let $m'$ be a monomial of degree
at most $e$. We say that $m'$ \emph{influences} $\tilde{m}\in
U_{s,e}(m_1)$ if $\alpha_{m'}$ appears with non-zero coefficient in
the equation corresponding to $\tilde{m}$. We now make the following claim.

\begin{claim}
\label{clm:UT}
Let $\tilde{m}\in U_{s,e}(m_1)$ and $m\in D_{s,e}(m_1)$ be such that
$\tilde{m} = m * m_1$. Then, $m$ influences $\tilde{m}$. Further, if
some monomial 
$m'$ influences $\tilde{m}$, then $m' \geq m$.
\end{claim}

Assuming the above claim, we complete the proof of the lemma as
follows. Consider the matrix $B$ of coefficients obtained by writing
the above linear system in the following manner. For each
$\tilde{m} = m*m_1\in U_{s,e}(m_1)$, we have a row of $B$ and let the
rows be arranged from top to bottom in increasing order of $m$ (w.r.t. the
graded lexicographic order). Similarly, for each $m'$ of degree at
most $e$, we have a column and again the columns are arranged from
left to right in increasing order of $m'$. The $(\tilde{m},m')$th
entry contains the coefficient of $\alpha_{m'}$ in the equation
corresponding to the coefficient of $\tilde{m}$.

Restricting our attention only to columns corresponding to
$m'\in D_{s,e}(m_1)$, \cref{clm:UT} guarantees to us that the
submatrix thus obtained is a $|D_{s,e}(m_1)|\times |D_{s,e}(m_1)|$
matrix that is upper triangular with non-zero entries along the
diagonal. Hence, the submatrix is full rank. In particular, the matrix
$B$ (and hence our linear system) has rank at least
$|D_{s,e}(m_1)|$. This proves the lemma.
\end{proof}

\begin{proof}[Proof of \cref{clm:UT}]
We start by showing that $m$ does indeed influence $\tilde{m}$. The linear equation corresponding to $\tilde{m}$ is 
\begin{equation}
\label{eq:tildemeq}
\sum_{(m',j):m'\cdot m_j = \tilde{m}} \beta_j \alpha_{m'} = 0
\end{equation}
where $m'$ runs over all monomials of degree at most $e$. 

Clearly, one of the summands in the LHS above is $\beta_1 \alpha_m$. Thus, to ensure that $m$ influences $\tilde{m}$, it suffices to ensure that no other summand containing the variable $\alpha_m$ appears. That is, that $m\cdot m_j \neq \tilde{m}$ for any $j > 1$. (Note that in general unique factorization is \emph{not true} in $\RM_q(n)$, since $X^q=X$.)

To see this, note further that $m\cdot m_j$ is either equal to $m*m_j$ (if they are disjoint) or has smaller degree than $m*m_j$. In either case, we have $m\cdot m_j \leq m*m_j$. Thus, we obtain
\[
m\cdot m_j \leq m* m_j < m*m_1 = \tilde{m}
\]
where the second inequality follows from the fact that $m_1 > m_j$ and
hence (by \cref{fac:mon-order}) $m'*m_1 > m'*m_j$ for any monomial $m'$. This shows that $\alpha_m$ appears precisely once in the left hand side of \eqref{eq:tildemeq} and in particular, that it must influence $\tilde{m}$.

Now, we show that no $m'< m$ influences $\tilde{m}$. Fix some $m' < m$. For any $j\in [N]$ we have
\[
m'\cdot m_j \leq m' * m_j \leq m' * m_1 < m* m_1 = \tilde{m}
\]
where the first two inequalities follow from a similar reasoning to above and the third from the fact that $m' < m$. Hence, we see that no monomial that is a product of $m'$ with another monomial from $f$ can equal $\tilde{m}$. In particular, this means that $m'$ cannot influence $\tilde{m}$.

This completes the proof of the claim.
\end{proof}

\begin{corollary}
\label{cor:fmaxmon-new}
Let $n,e,d,P,f$ be as in \cref{lem:fmaxmon}. Further, let $r$ be such that $(q-1)n - d = r$ and assume $r\geq 3e$. Then, 
$
\prob{P\sim \RM_q(n,e)}{\deg(fP) < d+e}\leq {q^{-N_q(\lfloor L/3\rfloor,e)}}
$
where $L = \lfloor r/(q-1)\rfloor.$
\end{corollary}

\begin{proof}
To prove the corollary, we use \cref{lem:fmaxmon} with $s = e$
and prove a lower bound on $|U_{e,e}(m_0)| = |U_e(m_0)| = |D_e(m_0)|$ where $m_0$ is the monomial from the statement of \cref{lem:Umax}. 

We first observe that we can assume that $r \geq 3(q-1).$ If this is not the case, then $\lfloor L/3\rfloor = 0$ and hence $N_q(\lfloor L/3\rfloor,e) = 1.$ Thus, the claimed bound on $\prob{P\sim \RM_q(n,e)}{\deg(fP) < d+e}$ follows from the fact that $|D_e(m_0)| \geq 1.$ Hence, we will assume from now on that $r \geq 3(q-1).$ In conjunction with our assumption that $r \geq 3e,$ this implies that 
\begin{equation}
    \label{eq:rlbd}
    r \geq 2e + (q-1).
\end{equation}

Let $T$ index the $L = \left\lfloor \frac{r}{q-1}\right\rfloor$ variables not present in the monomial $m_0$. We can lower bound $|D_e(m_0)|$ by the number of monomials of degree \emph{exactly} $e$ in $\RM_q(n,e)$ supported on variables from $T$; let $\mc{M}$ denote this set of monomials.

Partition $T$ arbitrarily into two sets $T_1$ and $T_2$ such that $|T_1| = L' = \lfloor L/3\rfloor$. 

To lower bound $|\mc{M}|$, note that given any monomial $m_1\in \RM_q(n,e)$ in the variables of $T_1$ of degree at most $e$, we can find a monomial $m_2$ over the variables of $T_2$ such that their product has degree exactly $e$. The reason for this is that the maximum degree of a monomial in the variables in $T_2$ is 
\ifIEEEtr
\else
\begin{align*}
(L-L')(q-1)
\geq \frac{L}{2}(q-1)\geq \frac{1}{2}(\frac{r}{q-1}-1)(q-1)
= \frac{r-(q-1)}{2} \geq e
\end{align*}
\fi
where the last inequality follows from \eqref{eq:rlbd}. Hence, we can always find a monomial $m_2$ over the variables in $T_2$ such that $\deg(m_1m_2) = e$. Hence, we can lower bound $|\mc{M}|$ by the number of monomials $m_1$ over the variables in $T_1$ of degree at most $e$ which is $N_q(L',e)$. We have thus shown that $|U_{e,e}(m_0)| \geq N_q(L',e)$. An application of \cref{lem:fmaxmon} now implies the corollary.
\end{proof}

\subsection{Connection to the Schwartz-Zippel Lemma over $\F_q$}
\label{sec:SZconn}

Consider the special case of \cref{lem:fmaxmon} when $e = (q-1)n$ and $s = 0$. In this case, note that $\RM_q(n,e)$ is just the ring $\RM_q(n)$ and hence the above lemma implies
$
\prob{P\sim \RM_q(n)}{\deg(fP) < d} \leq \frac{1}{q^{|U_{s,e}(m_0)|}}
$
where $m_0$ is the monomial from the statement of \cref{lem:Umax}. Note that as a special case, this implies that $\prob{P\sim \RM_q(n)}{fP = 0} \leq \frac{1}{q^{|U_{s,e}(m_0)|}}$.

Observe that  by \cref{fac:ring}, $fP = 0$ if and only if the polynomial $fP$ vanishes at each point of $\mathbb{F}_q^n$. However, since $P$ evaluates to an independent random value in $\mathbb{F}_q$ at each input $x\in \F_q^n$, we see that the probability that $fP$ evaluates to $0$ at each point is exactly the probability that $P(x) = 0$ at each point where $f(x)\neq 0$. This happens with probability exactly $\frac{1}{q^{|\Supp(f)|}}$. 

Putting it all together, we see that $\frac{1}{q^{|\Supp(f)|}} \leq \frac{1}{q^{|U_{s,e}(m_0)|}}$ and hence,
$
|\Supp(f)| \geq |U_{s,e}(m_0)| = |D_{s,e}(m_0)|.
$

For the chosen values of $e$ and $s$, the latter quantity is exactly the total number of monomials --- of \emph{any} degree --- that are disjoint from $m_0$, which is exactly $(q-v)q^{n-u-1}$, matching the Schwartz-Zippel lemma over $\F_q$ (\cref{fac:ring}).

It is also known that the Schwartz-Zippel lemma over $\F_q$ is tight for a suitably chosen degree $d$ polynomial $f$. \cref{lem:fmaxmon} is also tight for the same polynomial $f$, as we show below.

The Schwartz-Zippel lemma is tight for any $d \leq n(q-1)$ for the polynomial $f(X_1,\ldots, X_n)$ defined as follows. Write $d = u(q-1)+v$ so that $0\leq v < q-1$. Fix any ordering $\xi_0,\ldots,\xi_{q-1}$ of $\F_q$. Recall (see \cref{sec:basis}) that $\mc{B}_q(n,d)$ is the space of \emph{generalized monomials} w.r.t. this ordering of degree at most $d$. Let $f = b_{v}(X_{u+1})\cdot\prod_{i=1}^{u} b_{q-1}(X_i)$. Note that $f\in \mc{B}_q(n,d)$.

We show that this same $f$ also witnesses the tightness of \cref{lem:fmaxmon}. 

\begin{claim}
\label{clm:SZtight}
Let $f\in \RM_q(n)$ be as defined above. Then, for any $e,s\geq 0$ we have
\[
\prob{P\sim\RM_q(n,e)}{\deg(fP) < d+s}  = \frac{1}{q^{|U_{s,e}(m_0)|}}
\]
where $m_0$ is as defined in the statement of \cref{cor:Umax}.
\end{claim}

\begin{proof}
By \cref{lem:fmaxmon}, we already know that 
\[
\prob{P\sim\RM_q(n,e)}{\deg(fP) < d+s}  \leq \frac{1}{q^{|U_{s,e}(m_0)|}}.
\]
So it suffices to prove the opposite inequality. Namely that
\begin{equation}
\label{eq:lbd}
\prob{P\sim\RM_q(n,e)}{\deg(fP) < d+s}  \geq \frac{1}{q^{|U_{s,e}(m_0)|}}.
\end{equation}

For this proof, it is convenient to work with generalized monomials w.r.t. two different orderings. Consider the reverse ordering to the one defined above: i.e., $\xi_{q-1},\ldots,\xi_0$. Let $b_i'(X)$ denote the basis from \cref{sec:basis} w.r.t. this ordering. We define $\mc{B}'_q(n,e)$ to be the generalized monomials (see \cref{sec:basis}) w.r.t. this ordering of degree at most $e$.

We make a simple observation. Since each $b_i$ vanishes \emph{exactly} at $\xi_0,\ldots,\xi_{i-1}$ and each $b_j'$ vanishes exactly at $\xi_{q-1},\ldots, \xi_{q-j}$, we obtain
\begin{equation}
\label{eq:bibj'}
b_i(X)\cdot b_j'(X) = 0 \text{ iff $i + j \geq q$.}
\end{equation}

We say that $b_i$ and $b_j'$ are disjoint if $i+j < q$. Similarly, two generalized monomials $\prod_{i\in [n]}b_{j_i}(X_i)$ and $\prod_{i\in [n]}b_{j'_i}'(X_i)$ are disjoint if for each $i$, the basis elements $b_{j_i}$ and $b_{j_i'}'$ are disjoint. From \eqref{eq:bibj'} above, the product of any pair of non-disjoint generalized monomials with one from each of $\mc{B}_q(n,d)$ and $\mc{B}'_q(n,e)$ is $0$.

Since $\mc{B}'_q(n,e)$ forms a basis for $\RM_q(n,e)$ (\cref{fac:basis}), we can view the process of sampling  $P$ uniformly from $\RM_q(n,e)$ as picking $\alpha_{i_1,\ldots,i_n}\in \F_q$ independently and uniformly at random for each $(i_1,\ldots,i_n)$ such that $\sum_{j\in [n]}i_j \leq e$ and setting 
\[
P = \sum_{(i_1,\ldots,i_n): \sum_j i_j \leq e}\alpha_{i_1,\ldots,i_n}\prod_{j\in [n]} b_{i_j}'(X_j).
\]

We now consider the product $fP$, which is expanded as
\[fP = \sum_{(i_1,\ldots,i_n): \sum_j i_j \leq e}\alpha_{i_1,\ldots,i_n}f\cdot \prod_{j\in [n]} b_{i_j}'(X_j).
\]

From the definition of $f$ and using \eqref{eq:bibj'}, we see that the product of $f$ with each generalized monomial from $\mc{B}_q'(n,e)$ is non-zero if and only if
$i_j = 0$ for all $j \in [u]$ and $i_{u+1} + v < q$. In particular, the number of generalized monomials in $\mc{B}_q'(n,e)$ of degree exactly $t\leq e$ that are disjoint from $f$ is equal to the cardinality of the set
\[
D'_t(f)=\{(i_1,\ldots,i_n)\mid \sum_j i_j = t, i_j = 0\ \forall j \in [u], i_{u+1} + v < q\}
\]
By inspection, it is easily verified that the above set has the same cardinality as $D_t(m_0)$. In particular the size of the set $\bigcup_{s\leq t\leq e}D'_t(f)$ is $\sum_{s\leq t\leq e}|D'_t(f)|=|D_{s,e}(m_0)| = |U_{s,e}(m_0)|$. 

Note that when $\alpha_{i_1,\ldots,i_n} = 0$ for all $(i_1,\ldots,i_n) \in \bigcup_{s\leq t\leq e}D'_t(f)$, then we have
$\deg(fP) < d+s$. Since the coefficients $\alpha_{i_1,\ldots,i_n}$ are
chosen independently and uniformly at random from $\F_q$, this happens
with probability $q^{-|U_{s,e}(m_0)|}$. This implies \eqref{eq:lbd} and completes the proof of the claim.
\end{proof}

\section{Analyzing $\Test_{e,k}$}
\label{sec:testek}

We prove the main theorem of the paper, namely
\cref{thm:test-e-k-new}, in this section. The results
of this section only hold for \emph{prime} fields. 

For any non-negative integer parameters $L$ and $e$, recall that
$N_q(L,e)$ denotes the number of monomials $m$ in indeterminates
$X_1,\ldots,X_L$ such that the degree of each variable in $m$ is at
most $q-1$ and the total degree is at most $e$. Equivalently,
$N_q(L,e)$ is the dimension of the vector space $\RM_q(L,e).$ For $L <
0$, we define $N_q(L,e) = 1.$ 

We will choose the constant $c_q$ as in \cref{lem:HSS}.

We argue that the theorem holds by considering two cases.  We argue that when $f$ is $\Delta$-far from polynomials of degree $d + r/4$ --- a much stronger assumption than the hypothesis of the theorem --- then a modification of the proof of Dinur and Guruswami~\cite{DinurG2015} coupled with a suitable choice of basis for $\RM_q(n,d)$ yields the desired conclusion. 

If not, then $f$ is $\Delta$-close to some polynomial of degree
exactly $d'$ that is slightly larger than $d$. In this case, we can
argue that $f$ is ``essentially'' a polynomial of degree exactly $d'$ and for any such polynomial, the product $fP_1\ldots P_k$ is, w.h.p., a polynomial of degree exactly $d' + ek$ and hence $f\not\in\RM_q(n,d+ek)$. This requires the results of \cref{sec:SZext}.

We now proceed with the proof details. We consider the following two cases. 

\ifIEEEtr\else\begin{description}\fi
\ifIEEEtr
\else
\item[Case 1: $f$ is $\Delta$-far from $\RM_q(n,d+\frac{r}{4})$.]
\fi
See \cref{sec:case1} below.

\ifIEEEtr
\else
\item[Case 2: $f$ is $\Delta$-close to $\RM_q(n,d+\frac{r}{4})$.]
\fi 
Let $F\in \RM_q(n,d+\frac{r}{4})$ be such that $f$ is $\Delta$-close to $F$. Let $d' = \deg(F)$. Note that $d' > d$ since $f$ is $\Delta$-far from $\RM_q(n,d)$ by assumption. Hence, we must have $d < d' \leq d+\frac{r}{4}$. 

Note that for any $P_1,\ldots,P_k\in \RM_q(n,e)$, we have
$fP_1\cdots P_k$ is $\Delta$-close to $FP_1\cdots P_k$ (since $f(x) =
F(x)$ implies that $f(x)\cdot \prod_i P_i(x) = F(x)\cdot \prod_i
P_i(x)$). We have $FP_1\cdots P_k \in \RM_q(n,d'+ek) \subseteq \RM_q(n,d'+r/4) \subseteq \RM_q(n,d+r/2)$. Now if $fP_1\cdots P_k\in \RM_q(n,d+ek) \subseteq \RM_q(n,d+r/2)$, then by the Schwartz Zippel lemma over $\F_q$ (\cref{fac:ring}) applied to polynomials of degree at most $d+r/2$, we see that $fP_1\cdots P_k = FP_1\cdots P_k$. Hence, we have $FP_1\cdots P_k \in \RM_q(n,d+ek)$ which in particular implies that $FP_1\cdots P_k$ must have degree strictly less than $d'+ek$.

For this event to occur there must be some $i < k$ such that $FP_1\cdots P_i$ has degree exactly $d'_i := d' + ei$ but $FP_1\cdots P_{i+1}$ has degree strictly less than $d'_{i} + e$.

We have shown that
\begin{align}
\prob{P_1,\ldots,P_k}{fP_1\cdots P_k \in \RM_q(n,d+ek)}
&\ifIEEEtr\else\leq\fi \prob{P_1,\ldots,P_k}{\deg(FP_1\cdots P_k) < d'+ek}\notag\\
&\ifIEEEtr\else\leq\fi \sum_{i = 0}^{k-1} \prob{P_1\cdots
  P_k}{\deg\left(F\prod_{j=1}^{i+1}P_j \right) < d'_i + e\mid
  \deg\left(F\prod_{j=1}^i P_j\right) = d'_i}.\label{eq:case2eq1}
\end{align}

For each $i$, conditioning on any fixed choice of
$P_1,\ldots,P_{i}$, the right hand side of \eqref{eq:case2eq1} can
be bounded using \cref{cor:fmaxmon-new} applied with $d$ replaced by $d'_{i} \leq d
+ r/2-e= (q-1)n - (r/2+e)$ (the parameter $r/2+e$ satisfies the hypothesis of \cref{cor:fmaxmon-new} as $r \geq 4ek\geq 4e$ and hence $r/2+e\geq 3e$). The upper bound on the probability obtained from \cref{cor:fmaxmon-new} is $q^{-N_q(\lfloor L'/3\rfloor,e)}$ where $L' = \lfloor r/(q-1)\rfloor.$ Note that by our assumption that $\Delta \leq q^{r/4(q-1)-2},$ we have $\lfloor L'/3\rfloor \geq L = \lfloor \log_q\Delta\rfloor.$ Hence, using \eqref{eq:case2eq1} we have 
\[
\prob{P_1,\ldots,P_k}{fP_1\cdots P_k \in \RM_q(n,d+ek)} \leq k q^{-N_q(L,e)}.
\]

This implies \cref{thm:test-e-k-new} in this case.
\ifIEEEtr\else\end{description}\fi

\subsection{Case 1 of \cref{thm:test-e-k-new}: $f$ is $\Delta$-far from $\RM_q(n,d+\frac{r}{4})$}
\label{sec:case1}

In this case, we adopt the method of Dinur and Guruswami~\cite{DinurG2015} along with a suitable choice of basis (\cref{sec:basis}) and \cref{lem:multilineqns} to bound the required probability. The proof is an induction, the key technical component of which is \cref{lem:HSS}, which follows from the work of Haramaty et al.~\cite{HaramatySS2013}. 

Let $d' = d+r/4$. Since we know that $f$ is not of degree $d'$ (indeed
it is $\Delta$-far from $\RM_q(n,d')$), we intuitively believe that
$fP_1\cdots P_k$ should not even belong to
$\RM_q(n,d'+ek)\supsetneq \RM_q(n,d+ek)$. Hence, we associate
with the event that $fP_1\cdots P_k\in \RM_q(n,d+ek)$ the
``surprise'' parameter $s := d'-d$. This will be one of the
parameters we will track in the induction. Recall that for our setting
of parameters $s = r/4 \geq ek$. 

\begin{definition}
\label{def:rho}
For any positive integers $n_1,r_1,\Delta_1,$ $e_1\geq 0,$ and $s_1\geq e_1k$, we define the quantity $\rho(n_1,e_1,r_1,\Delta_1,s_1)$ to be the largest $\rho\in \mathbb{R}$ such that for any $d_1\geq 0$ such that $d_1 \leq (q-1)n_1-s_1-r_1$ and for any $f$ that is $\Delta_1$-far from $\RM_q(n_1,d_1+s_1)$ for $0 < \Delta_1 < q^{r_1/(q-1)}$, we have
\[
\prob{P_1,\ldots,P_k\sim \RM_q(n_1,e_1)}{fP_1\cdots P_k \in \RM_q(n_1,d_1+e_1k)} \leq q^{-\rho}.
\]
\end{definition}

We prove by induction on $e_1,r_1,$ and $\Delta_1$ that for any $n_1,e_1,r_1,\Delta_1,s_1$ as above,
\begin{equation}
\label{eq:rho-lbd}
\rho(n_1,e_1,r_1,\Delta_1,s_1) \geq \eta(q,k)\cdot N_q(\lfloor\frac{L_1}{10}\rfloor-c_q,e_1)
\end{equation}
where $\eta(q,k)$ is as in the statement of the theorem, $L_1 = \lfloor \log_q \Delta_1 \rfloor$, and $c_q$ is as defined in \cref{lem:HSS}. Note that applying \eqref{eq:rho-lbd} with $n_1 =n, e_1=e,r_1 =r,\Delta_1 = \Delta$ and $s_1=s$ immediately implies the result of this section (i.e. the statement of Theorem~\ref{thm:test-e-k-new} in this case).
%

The base case of the induction --- which we apply when either $e_1 = 0$, $r_1 \leq c_q$, or $\Delta_1 \leq q^5$ --- is the following simple lemma. (It is stated in greater generality than needed in the rest of the proof.)

\begin{lemma}
\label{lem:case1base}
For any positive $n_1,r_1,$ and $\Delta_1$; $e_1 \geq 0$; and $s_1\geq e_1 k$, we have $\rho(n_1,e_1,r_1,\Delta_1,s_1) \geq \eta(q,k)$.
\end{lemma}

The inductive case is captured in the following lemma.

\begin{lemma}
\label{lem:case1indn}
For any positive $n_1,e_1,r_1,\Delta_1$ and $s_1\geq e_1 k$ with $e_1 > 0$, $r_1 \geq c_q$ and $q^5< \Delta_1 < q^{r_1/(q-1)}$, we have 
\ifIEEEtr
\else
\[
\rho(n_1,e_1,r_1,\Delta_1,s_1) \geq \sum_{i=0}^{\min\{e_1,q-1\}} \rho(n_1-1,e_1 - i,r_1 - (q-1), \Delta_1/q^3, s_1- ki).
\]
\fi
\end{lemma}

Assuming both these lemmas, we can quickly finish the proof of \eqref{eq:rho-lbd} as follows. We proceed by induction on $e_1+\Delta_1 + r_1.$ In case either $e_1=0$ or $\Delta_1 \leq q^5$ or $r_1 < c_q$, we can easily infer \eqref{eq:rho-lbd} using \cref{lem:case1base} and using the fact that $N_q(\lfloor L_1/10 \rfloor-c_q,e_1) = 1$. This is by a simple case analysis.
\begin{itemize}
\item Assume $e_1=0.$ In this case, $N_q(\lfloor L_1/10 \rfloor-c_q,e_1) = 1$ since either $\lfloor L_1/10 \rfloor-c_q \geq 0$ and hence the number of monomials of degree at most $e_1$ in $\lfloor L_1/10 \rfloor-c_q$ many variables is $1$, or $\lfloor L_1/10 \rfloor-c_q < 0$ and $N_q(\lfloor L_1/10 \rfloor-c_q,e_1) = 1$ by definition.
\item Now assume that $\Delta_1 \leq q^5.$ In this case, we see immediately that $\lfloor L_1/10 \rfloor-c_q < 0$ and hence $N_q(\lfloor L_1/10 \rfloor-c_q,e_1) = 1$ by definition.
\item Finally assume that $r_1 < c_q$. In this case, $L_1 = \lfloor \log_q \Delta_1\rfloor <  r_1/(q-1) < c_q.$ Hence, we again have $\lfloor L_1/10 \rfloor-c_q < 0$ and thus $N_q(\lfloor L_1/10 \rfloor-c_q,e_1) = 1$ by definition.
\end{itemize}

The above proves the base case of the induction. For the inductive case when all the hypotheses of \cref{lem:case1indn} hold, we see that 
\begin{align*}
\rho(n_1,e_1,r_1,\Delta_1,s_1) 
&\geq \sum_{i=0}^{\min\{e_1,q-1\}} \rho(n_1-1,e_1 - i,r_1 - (q-1), \Delta_1/q^3, s_1- ki)\\
&\geq \eta(q,k)\cdot \sum_{i=0}^{\min\{e_1,q-1\}} N_q(\lfloor (L_1 - 3)/10 \rfloor-c_q,e_1-i)\\
&\geq \eta(q,k)\cdot \sum_{i=0}^{\min\{e_1,q-1\}} N_q(\lfloor L_1/10 \rfloor-1-c_q,e_1-i)\\
&\geq \eta(q,k)\cdot N_q(\lfloor L_1/10 \rfloor-c_q,e_1),
\end{align*}
where the first inequality is simply the statement of \cref{lem:case1indn}, the second follows by induction, and the fourth follows from the simple observation that for any $L'\in \mathbb{Z}$ and $e' > 0,$
\[
N_q(L',e') \leq \sum_{i=0}^{\min\{e',q-1\}} N_q(L'-1,e'-i).
\] 
This finishes the proof of \eqref{eq:rho-lbd} assuming \cref{lem:case1base} and \cref{lem:case1indn}. We now prove these lemmas.

\begin{proof}[Proof of \cref{lem:case1base}]
Fix any $d_1 \leq (q-1)n_1 - s_1-r_1$ and any $f\in \RM_q(n_1)$ that is $\Delta_1$-far from $\RM_q(n_1,d_1+s_1)$. In particular, $f\not\in \RM_q(n_1,d_1)$. Say $f$ is of degree $d'$ for some $d' > d_1$. As we have $d_1+e_1k\leq d_1+s_1 < (q-1)n_1$, we can fix some $d''$ such that $d_1+e_1k < d'' \leq \min\{(q-1)n_1,d'+e_1k\}$. 

We first show that there exists a monomial $m$ of degree $d''$ and a choice for $P_1,\ldots,P_k$ such that the monomial $m$ has non-zero coefficient in $fP_1\cdots P_k$. If $d'' = d'$, then we can take $m$ to be any monomial of degree $d'$ with non-zero coefficient in $f$ and $P_1,\ldots,P_k$ to each be the constant polynomial $1$. Otherwise, let $d'' = d'+\delta$; note that $\delta\leq e_1k$. Let $\tilde{m} = \LM(f)$ (of degree $d'$). We choose any $m'\in D_{\delta}(\tilde{m})$. Since $\deg(m') = \delta \leq e_1k$, we can find $m_1',\ldots,m_k'$ of degrees at most $e_1$ each such that $m' = m_1'\cdots m_k'$. We set $m = \tilde{m}m'$. It can be checked that if $P_1 = m_1',\ldots,P_k = m_k'$, then the monomial $m$ appears with non-zero coefficient in $fP_1\cdots P_k = fm'$.

We now consider the probability that $m$ has a non-zero coefficient in the random polynomial $g = fP_1\cdots P_k$ obtained when each $P_i$ is chosen uniformly from $\RM_q(n_1,e_1)$. The coefficient of $m$ in $g$ can be seen to be a polynomial $R$ of degree at most $k$ in the coefficients of $P_1,\ldots,P_k$. Since we have seen above that there is a choice of $P_1,\ldots,P_k$ such that this coefficient is non-zero, we know that $R$ is a non-zero polynomial. By the Schwartz-Zippel lemma (\cref{fac:ring}), we see that the probability that $R$ is non-zero is at least $q^{-k/(q-1)}$. Thus, with probability at least $q^{-k/(q-1)}$, the monomial $m$ has non-zero coefficient in $g$ and hence $\deg(g) \geq d'' > d_1+e_1k$. 

Hence, the probability that $\deg(g) \leq d_1+e_1k$ is upper bounded by $(1-q^{-k/(q-1)})$. Using the standard inequality $1-x\leq \exp(-x)$ and the definition of $\eta(q,k)$, we see that
\[
\prob{P_1,\ldots,P_k}{\deg(g) \leq d_1+e_1k} \leq \exp(-\frac{1}{q^{k/(q-1)}}) \leq q^{-\eta(q,k)}.
\]
This proves the lemma.
\end{proof}

\begin{proof}[Proof of \cref{lem:case1indn}]
Fix any $d_1 \leq (q-1)n_1 - s_1 - r_1$ and any $f\in \RM_q(n_1)$
that is $\Delta_1$-far from $\RM_q(n_1,d_1+s_1)$. Since $r_1 \geq
c_q$, \cref{lem:HSS} is applicable to $f$. Hence, there is a linear function $\ell(X)$ such that for each $\alpha \in \F_q$, the restricted function $f|_{\ell(X)=\alpha}$ is $\Delta_1/q^3$-far from $\RM_q(n_1-1,d_1+s_1)$. By applying a linear transformation to the set of variables, we may assume that $\ell(X) = X_{n_1}$. 

Let $q' = \min\{e_1,q-1\}.$ Note that $q' >0.$

Fix any ordering $\{\xi_0,\ldots,\xi_{q-1}\}$ of the field $\F_q$ and consider the univariate basis polynomials $b_i(X)$ ($0\leq j < q$) w.r.t. this ordering  as defined in \cref{sec:basis}. We can view the process of sampling each $P_i(X_1,\ldots,X_{n_1})\in \RM_q(n_1,e_1)$ as independently sampling $Q_{i,j}(X_1,\ldots,X_{n_1-1})\in \RM_q(n_1-1,e_1-j)$ ($0\leq j \leq q'$) and setting $P_i = \sum_{0\leq j < q}b_j(X_{n_1}) Q_{i,j}(X_1,\ldots,X_{n_1-1})$ where $Q_{i,j} = 0$ for $j\in \{q'+1,\ldots,q-1\}$. Let $P$ denote $P_1\cdots P_k$. We can also decompose $P = \sum_{0\leq j < q} b_j(X_{n_1}) Q_j(X_1,\ldots,X_{n_1-1})$.

We now use \cref{lem:ut-decomp}, by which can decompose the product $fP$ as follows
\begin{equation}
\label{eq:fP-decomp}
fP = \sum_{\ell=0}^{q-1}b_\ell(X_{n_1}) \left(Q_\ell\cdot f|_{X_{n_1}=\xi_\ell} + \sum_{0\leq j<\ell}Q_j\cdot h_{j,\ell}\right)
\end{equation}
where each $h_{j,\ell}(X_1,\ldots,X_{n_1-1})$ is some element of $\RM_q(n_1-1)$.

By \cref{lem:prod-basis}, it follows that for each $\ell < q$
\begin{equation}
\label{eq:prod-qij}
Q_\ell = \sum_{(\ell_1,\ldots,\ell_k) \leq \ell}\beta^{(\ell)}_{(\ell_1,\ldots,\ell_k)} Q_{(\ell_1,\ldots,\ell_k)}
\end{equation}
where $\beta^{(\ell)}_{(\ell,\ldots,\ell)} \neq 0$ and $Q_{(\ell_1,\ldots,\ell_k)} = \prod_{i\in [k]}Q_{i,\ell_i}$. Plugging \eqref{eq:prod-qij} into \eqref{eq:fP-decomp} we obtain
\ifIEEEtr
\else
\begin{align}
fP &= \sum_{\ell=0}^{q-1}b_\ell(X_{n_1})\cdot \left(f|_{X_{n_1}=\xi_\ell}\sum_{(\ell_1,\ldots,\ell_k) \leq \ell}\beta^{(\ell)}_{(\ell_1,\ldots,\ell_k)} Q_{(\ell_1,\ldots,\ell_k)}+ \sum_{0\leq j<\ell} h_{j,\ell} \sum_{(\ell_1,\ldots,\ell_k)\leq j}\beta^{(j)}_{(\ell_1,\ldots,\ell_k)} Q_{(\ell_1,\ldots,\ell_k)}\right)\notag\\
&= \sum_{\ell=0}^{q-1}b_\ell(X_{n_1})\times \underbrace{\left(\beta^{(\ell)}_{(\ell,\ldots,\ell)}Q_{(\ell,\ldots,\ell)} f|_{X_{n_1}=\xi_\ell}
+  \sum_{(\ell_1,\ldots,\ell_k)< \ell}Q_{(\ell_1,\ldots,\ell_k)}h^{(\ell)}_{(\ell_1,\ldots,\ell_k)}\right)}_{:= R_\ell(X_1,\ldots,X_{n_1-1})}\label{eq:fP-decomp-2}
\end{align}
\fi
where each $h^{(\ell)}_{(\ell_1,\ldots,\ell_k)} = h^{(\ell)}_{(\ell_1,\ldots,\ell_k)}(X_1,\ldots,X_{n_1-1})\in \RM_q(n_1-1)$. We also use $h^{(\ell)}_{(\ell,\dots,\ell)}$ to denote $\beta^{(\ell)}_{(\ell,\ldots,\ell)}f|_{X_{n_1} = \xi_\ell}$.

Now, we analyze the probability that $fP\in \RM_q(n_1,d_1+e_1k)$. We have
\begin{align}
\prob{Q_{i,j}}{fP\in \RM_q(n_1,d_1+e_1k)} 
 &\leq \prob{Q_{i,j}}{\bigwedge_{0\leq \ell < q}R_\ell\in \RM_q(n_1,d_1+e_1k - \ell)}\notag\\
&\leq \prod_{0\leq \ell < q}\prob{Q_{i,j}}{R_\ell\in \RM_q(n_1,d_1+e_1k - \ell)\mid \{R_0,\ldots,R_{\ell-1}\}}\notag\\
&\leq \prod_{0\leq \ell < q}\prob{Q_{i,j}}{R_\ell\in \RM_q(n_1,d_1+e_1k - \ell)\mid \{Q_{i,j}\mid i\in [k],j<\ell\}}\label{eq:cond-prob}
\end{align}
where the last inequality follows from the fact that each $R_j$ only depends on $Q_{i,j'}$ where $i\in [k]$ and $j'\leq j$.

Let $\exp_q(\theta)$ denote $q^{\theta}$. We claim that for each $\ell\in \{0,\ldots,q'\}$, the $\ell$th term in the RHS of \eqref{eq:cond-prob} can be bounded as follows.
\ifIEEEtr
\else
\begin{equation}
\label{eq:Rell}
\prob{Q_{i,j}}{R_\ell\in \RM_q(n_1,d_1+e_1k - \ell)\mid \{Q_{i,j}\mid i\in [k],j<\ell\}} \leq \exp_q(-\rho(n_1-1,e_1-\ell,r_1-(q-1),\Delta_1/q^3,s_1-k\ell))
\end{equation}
\fi

Substituting into \eqref{eq:cond-prob} (and using the trivial upper bound of $1$ for terms corresponding to $\ell\in \{q'+1,\ldots,q-1\}$) this will show that 
\ifIEEEtr
\else
\[
\rho(n_1,e_1,r_1,\Delta_1,s_1) \geq \sum_{\ell=0}^{q'} \rho(n_1-1,e_1 - \ell,r_1-(q-1), \Delta_1/q^3, s_1- k\ell)
\]
\fi
which proves the lemma.

It remains only to prove \eqref{eq:Rell} for which we use \cref{lem:multilineqns}. We first condition on any choice of $Q_{i,j}$ for $i\in [k]$ and $j < \ell$. The event $R_\ell\in \RM_q(n_1-1,d_1+e_1k-\ell)$ now depends only on the random polynomials in $\mc{Q} = \{Q_{i,\ell}\mid i\in [k]\}$. We view the process of sampling these polynomials as sampling the coefficients of the standard monomials $m\in \RM_q(n_1-1,e_1-\ell)$\footnote{Any basis for the space $\RM_q(n_1-1,e_1-\ell)$ will do here. In particular, we do not need the special basis from \cref{sec:basis}.} independently and uniformly at random from $\F_q$. Let $\zeta_{i,m}$ denote the (random) coefficient of the monomial $m$ in the polynomial $Q_{i,\ell}$.

Scanning the definition of $R_\ell$ in \eqref{eq:fP-decomp-2} above, we see that $R_\ell$ is the sum of polynomials $Q_{(\ell_1,\ldots,\ell_k)}h^{(\ell)}_{(\ell_1,\ldots,\ell_k)}$, where $(\ell_1,\ldots,\ell_k) \leq \ell$. For each $(\ell_1,\ldots,\ell_k) < \ell$, the polynomial $Q_{(\ell_1,\ldots,\ell_k)}$ is a product of at most $k-1$ polynomials from the set $\mc{Q}$.

The event that $R_\ell\in \RM_q(n_1-1,d_1+e_1k-\ell)$ is equal to the probability that each monomial $\tilde{m}$ of degree larger than $d_1+e_1k-\ell$ has zero coefficient in $R_\ell$. Consider the coefficient of $\tilde{m}$ in each term 

\begin{equation}
\label{eq:Qell}
h^{(\ell)}_{(\ell_1,\ldots,\ell_k)}Q_{(\ell_1,\ldots,\ell_k)} = Q'_{(\ell_1,\ldots,\ell_k)}\prod_{i:\ell_i = \ell}Q_{i,\ell_i}
\end{equation}
where $Q'_{(\ell_1,\ldots,\ell_k)}$ is the \emph{fixed}  polynomial $\prod_{i:\ell_i < \ell}Q_{i,\ell_i}\cdot h^{(\ell)}_{(\ell_1,\ldots,\ell_k)}$.

Let $\mc{Z} = \{\zeta_{i,m}\ \mid  i\in [k], m\in \RM_q(n_1-1,e_1-\ell)\}$ and $\mc{Z}_i = \{\zeta_{i,m}\ \mid  m\in \RM_q(n_1-1,e_1-\ell)\}$ for each $i\in [k]$. Clearly, $\Pi = \{\mc{Z}_1,\ldots,\mc{Z}_k\}$ is a partition of $\mc{Z}$. It can be verified from \eqref{eq:Qell} that the coefficient of each monomial $\tilde{m}$ in $h^{(\ell)}_{(\ell_1,\ldots,\ell_k))}Q_{(\ell_1,\ldots,\ell_k)}$ is a $\Pi$-multilinear  polynomial (see \cref{sec:multilin}) $C_{(\ell_1,\ldots,\ell_k)}^{(\tilde{m})}$ applied to the random variables in $\mc{Z}$. In fact, it only depends on the random variables in $\bigcup_{i:\ell_i = \ell}\mc{Z}_i$.  Hence, this polynomial is $\Pi$-set-multilinear if and only if $\ell_1 = \dots = \ell_k = \ell$. 

Hence, from the definition of $R_\ell$ \eqref{eq:fP-decomp-2} we see that the coefficient of $\tilde{m}$ in $R_\ell$ is 
\begin{equation}
\label{eq:coeff}
C^{(\tilde{m})} := \sum_{(\ell_1,\ldots,\ell_k)\leq \ell}C^{(\tilde{m})}_{(\ell_1,\ldots,\ell_k)}
\end{equation}
which is a $\Pi$-multilinear polynomial in $\mc{Z}$ with
set-multilinear part $C^{(\tilde{m})}_{(\ell,\ldots,\ell)}$. We will
use \cref{lem:multilineqns} to bound the probability that $C^{(\tilde{m})}(\zeta_{i,m}: i,m) = 0$.

Now we can analyze the probability that $R_\ell\in \RM_q(n_1-1,d_1+e_1k-\ell)$. We omit the conditioning on $Q_{i,j}$ ($j < \ell$) since they are fixed. Below, $\tilde{m}$ varies over all monomials in $\RM_q(n_1-1)$ of degree $> d_1 + e_1k -\ell$.
\begin{align}
\prob{Q_{i,\ell}}{R_\ell\in \RM_q(n_1-1,d_1+e_1k-\ell)}  
&= \prob{\zeta_{i,m}}{\bigwedge_{\tilde{m}}C^{(\tilde{m})}(\zeta_{i,m}) = 0}\notag\\
&\leq \prob{\zeta_{i,m}}{\bigwedge_{\tilde{m}}C^{(\tilde{m})}_{(\ell,\ldots,\ell)}(\zeta_{i,m}) = 0}\notag\\
&=\prob{\zeta_{i,m}}{Q_{(\ell,\ldots,\ell)}h^{(\ell)}_{(\ell,\ldots,\ell)} \in \RM_q(n_1-1,d_1+e_1k-\ell)}\notag\\
&= \prob{\zeta_{i,m}}{Q_{(\ell,\ldots,\ell)}f|_{X_{n_1} = \xi_\ell}\in \RM_q(n_1-1,d_1+e_1k-\ell)} \label{eq:smpart}
\end{align}
where the inequality follows from \cref{lem:multilineqns}; the second equality follows from the fact that $C^{(\tilde{m})}_{(\ell,\ldots,\ell)}(\zeta_{i,m}) = 0$ for all $\tilde{m}$ if and only if each monomial of degree more than $d_1+e_1k-\ell$ has zero coefficient in $Q_{(\ell,\ldots,\ell)}h^{(\ell)}_{(\ell,\ldots,\ell)}$; and the last equality follows from the fact that $h^{(\ell)}_{(\ell,\ldots,\ell)} = \beta^{(\ell)}_{(\ell,\ldots,\ell)}f|_{X_{n_1} = \xi_\ell}$ and $\beta^{(\ell)}_{(\ell,\ldots,\ell)}\neq 0$.

The final expression in \eqref{eq:smpart} can be bounded by the induction hypothesis applied with $n_2 = n_1-1$, $e_2 = e_1 - \ell$, $r_2 = r_1 -(q-1)$, $\Delta_2 = \Delta_1/q^3$ and $s_2 = s_1- k\ell$. We show below that the parameters satisfy all the required conditions from \cref{def:rho}.
\begin{itemize}
\item Note that $r_2 = r_1 - (q-1) > 0$ as $r_1 \geq c_q > q$ (see \cref{lem:HSS} for the final inequality). 
\item $Q_{(\ell,\ldots,\ell)} = \prod_{i}Q_{i,\ell}$ is a product of $\ell$ polynomials independently and uniformly sampled from $\RM_q(n_1-1,e_1 - \ell)=\RM_q(n_2,e_2)$. Recall that $e_1 \geq q'\geq \ell$ and hence $e_2=e_1-\ell \geq 0$.
\item By assumption, $g:=f|_{X_n = \xi_\ell}$ is $\Delta_1/q^3 = \Delta_2$-far from $\RM_q(n_1-1,d_1+s_1) = \RM_q(n_2,d_2+s_2)$ where $d_2 = d_1 + k\ell$ and $s_2$ is as defined above. Note that $s_2 = s_1 - k\ell \geq e_1k-k\ell = e_2k$. Also note that 
\ifIEEEtr
\else
\[(q-1)n_2 - d_2 = (q-1)n_1 - (q-1) - d_1 - k\ell \geq r_1 + s_1 - (q-1) - k\ell = r_2 + s_2,\]
\fi
where the inequality uses $d_1 \leq (q-1)n_1 - r_1 - s_1$. Hence, we have $d_2 \leq (q-1)n_2 - r_2 - s_2.$
\item We also have $\Delta_2 = \Delta_1/q^3 < q^{r_1/(q-1) - 3} < q^{r_2/(q-1)}$. Similarly, as $\Delta_1 > q^5,$ we have $\Delta_2 > 0.$
\item Finally, we consider the event that $g\prod_{i}Q_{i,\ell}\in \RM_q(n_1-1,d_1+e_1k-\ell)= \RM_q(n_2,d_2 + e_2k-\ell) \subseteq \RM_q(n_2, d_2 + e_2k)$.
\end{itemize}

Thus, we can upper bound the probability in \eqref{eq:smpart} by $\exp_q(-\rho(n_2,e_2,r_2,\Delta_2,s_2))$, which yields \eqref{eq:Rell} and proves the lemma.
\end{proof}

\section{Two applications}
\label{sec:applns}

\subsection{A question of Dinur and Guruswami}

In this section, we show how \cref{thm:test-e-k} implies
\cref{lem:robustdg}, thus answering a open question raised by
Dinur and Guruswami~\cite{DinurG2015}.

\begin{proof}[Proof of \cref{lem:robustdg}]
The proof of the lemma for robustness $\Delta'$ can be reduced to
\cref{thm:test-e-k} for $k = 2$ as follows.

Let $f$ be $\Delta$-far from $\RM_q(n,d)$ as stated in the
lemma. Call $P$ ``lucky'' if
$\Delta(f\cdot P, \RM_q(m,d+e))\leq \Delta'$. We need to bound the
probability $\Pr_{P \in \RM_q(n,e)}[P \text{ is lucky }]$. For a lucky $P$, let $F$ be a degree-$(d+e)$
polynomial that is $\Delta'$-close to $f\cdot P$. Now, choose
$P'\in_R \RM_q(n,e)$ and let
$g = f P\cdot P'$. Also, let
$G = F\cdot P'$; note that $G\in \RM_q(n,d+2e)$.

Let $D = \{x\in \F_q^n\ \mid \ F(x)\neq f(x)P(x)\}$. We have $|D|\leq \Delta'$. Further, if $P'(x) = 0$ for each $x\in D$, then we have $g = G$ and hence $g\in \RM_q(n,d+2e)$.

Observe that the event that $P'(x) = 0$ for each $x\in D$ is a set of
$|D|\leq \Delta'$ homogeneous linear equations in the (randomly
chosen) coefficients of $P$. These equations simultaneously vanish
with probability at least $q^{-\Delta'}$. Hence, for a lucky $P$, we
see that $\prob{P'}{g\in \RM_q(n,d+2e)}\geq q^{-\Delta'}$. 

Thus, we see that for independent and randomly chosen $P,P'\in \RM_q(n,e)$,
\begin{align*}
&\prob{P,P'}{fPP' \in \RM_q(n,d+2e)}\\
& \geq \Pr_P[P \text{ is lucky }] \cdot \Pr_{P,P'}[ g
  \in \RM_q(n,d+2e) \mid P \text{ is lucky }]\\
& \geq \Pr_P[P \text{ is lucky }] \cdot \Pr_{P,P'}[ g
  = G \mid  P \text{ is lucky }]\\
& \geq \Pr_P[P \text{ is lucky }]\cdot \frac{1}{q^{\Delta'}}.
\end{align*}

Thus, by \cref{thm:test-e-k} we get
\[
\Pr_P[P \text{ is lucky }]\leq \frac{q^{\Delta'}}{q^{q^{\Omega(r)}}}.
\]

The lemma now follows for some $\Delta' = q^{\Omega(r)}$.
\end{proof}

\subsection{Analysis of $\text{Corr-$h$}$}

Recall the test $\text{Corr-$h$}$ defined in the introduction where $h
\in \RM_q(n,k)$ is a polynomial of exact degree $k$. In this
section, we analyze this test $\text{Corr-$h$}$, thus proving
\cref{cor:corrq}. 

For this we need the following two properties of polynomials.

\newcommand{\oexp}[1]{\omega^{\langle #1\rangle}}
\ifIEEEtr
\paragraph{Dual of $\RM_q(n,d)$}
\else
\begin{description}
\item[Dual of $\RM_q(n,d)$:]
\fi 
For any two functions, $f,g \in
  \mc{F}_q(n)$, define $\langle f, g \rangle : = \sum_{x \in \F_q^n}
  f(x) \cdot g(x)$. Given any $\F_q$-space $\mc{C} \subseteq
  \mc{F}_q(n)$, the dual of $\mc{C}$ is defined as $\mc{C}^\perp:= \{
  f \in \mc{F}_q(n) \mid \forall g \in \mc{C}, \langle f, g \rangle =
  0\}$. Recall that $r = (q-1)n-d$. It is well-know that the sets of polynomials
  $\RM_q(n,d)$ and $\RM_q(n,r-1)$ are
  duals of each other~\cite{vanLint}. We use these dual spaces to
  write the indicator variable for the event ``$f \in
\RM_q(n,d)$'' equivalently as
$\mathbbm{1}_{f \in \RM_q(n,d)} = \avg{Q \in
  \RM_q(n,r-1)}{\oexp{f,Q}},$
where $\omega = e^{2\pi i/q}$. This follows from the following
observations. 
\begin{itemize}
\item For any polynomial $P \in \RM_q(n,d)$, we have that for all $Q \in
  \RM_q(n,r-1)$, $\langle P, Q \rangle = 0$. Thus, in this case we
  have $\avg{Q \in
  \RM_q(n,r-1)}{\oexp{P,Q}} =1$.
\item Let $f \notin \RM_q(n,d)$. For each $\alpha \in \F_q$,
  let  $\mc{C}_\alpha: = \{ Q \in
  \RM_q(n,r-1) \mid \langle f,Q \rangle = \alpha \}$. Since $f
  \notin \RM_q(n,d)$, there exists a $Q \in
  \RM_q(n,r-1)$ such that $\langle f, Q \rangle \neq 0$ and hence $\mc{C}_0$ is a proper subspace of
  $\RM_q(n,r-1)$. This implies that $\{\mc{C}_\alpha\}_{\alpha \in \F_q}$
  form an equipartition of $\RM_q(n,r-1)$. Hence, $\avg{Q \in
  \RM_q(n,r-1)}{\oexp{f,Q}} = \avg{\alpha \in \F_q}{\avg{Q \in
  \mc{C}_\alpha}{\oexp{f,Q}}} = \avg{\alpha \in \F_q}{\omega^\alpha} =
0$.  
\end{itemize}

\ifIEEEtr
\else
\item[Squaring trick:]
\fi
We use a standard squaring trick to bound the absolute
  value of the quantity $\avg{P}{\oexp{h(P), f}}$. Let $g$ be a
  univariate polynomial of degree exactly $k$ with leading coefficient
  $g_k$. We will show (using induction on $k$) that for all $k \geq
  1$, we have 
\begin{align*}
\left|\avg{P}{\oexp{g(P),f}}\right|^{2^k} & \leq
                                            \avg{P_1,\ldots,P_k}{\oexp{k!g_k
                                            P_1\cdots P_k,f}}.
\end{align*}
The base case of the induction $(k=1)$ can be easily checked to be
true. Let $g(P) = aP + b$ where $a \neq 0$.
\ifIEEEtr
\else
\begin{align*}
\left|\avg{P}{\oexp{aP+b,f}}\right|^2 
& = \avg{P,P_1}{\oexp{(a(P+P_1)+b),f}\cdot \oexp{- (aP+b),
  f}}
=\avg{P,P_1}{\oexp{aP_1,f}} = \avg{P_1}{\oexp{aP_1,f}}.
\end{align*}
\fi
We now induct from $k-1$ to $k$. Let $g$ be a polynomial of degree
exactly $k$ with leading coefficient $g_k$. To this end, we first observe that
$g(P+P_1)-g(P)$ is a polynomial of degree exactly $k-1$ in $P$ with leading
coefficient $kP_1g_k$.
\ifIEEEtr
\else
\begin{align*}
\left|\avg{P}{\oexp{g(P),f}}\right|^{2^k} 
&=\left(\left|\avg{P}{\oexp{g(P),f}}\right|^{2}\right)^{2^{k-1}}= \left(\avg{P,P_1}{\oexp{g(P+P_1)-g(P),f}}\right)^{2^{k-1}}\\
\text{(by convexity)} &\leq
  \avg{P_1}{\left|\avg{P}{\oexp{g(P+P_1)-g(P),f}}\right|^{2^{k-1}}}\\
\text{(by induction)} & \leq \avg{P_1}{\avg{P_2,\dots,P_k}{\oexp{(k-1)!\cdot (kP_1g_k) \cdot
  P_2P_3\cdots P_k,f}}} =\avg{P_1,\ldots,P_k}{\oexp{k!g_k
                                            P_1\cdots P_k,f}}.
\end{align*}
\fi
\ifIEEEtr\else\end{description}\fi

We are now ready to prove \cref{cor:corrq}.
\begin{proof}[Proof of \cref{cor:corrq}] Since the class of polynomials $\RM_q(n,d+ek)$ is closed under scalar multiplication, we can assume (by multiplying by a non-zero scalar if necessary) that $h$ is monic. 
\ifIEEEtr
\else
\begin{align*}
\Pr_{P\in \RM_q(n,e)}\left[ f \cdot h(P) \in \RM_q(n,d+ek)\right ]
& = \left|\avg{P \in \RM_q(n,e), Q \in \RM_q(n,s-1)}{\oexp{f
  \cdot h(P), Q}} \right| \left| \avg{Q}{\avg{P}{\oexp{h(P),fQ}}}
  \right|^{2^k/2^k} \\
\text{(by convexity)}&\leq 
  \left(\avg{Q}{\left|\avg{P}{\oexp{h(P),fQ}}\right|^{2^k}}\right)^{1/2^k}\\
\text{(by the squaring trick)}&\leq \left(
  \avg{Q}{\avg{P_1,\ldots,P_k}{\oexp{k! P_1\cdots
  P_k,fQ}}}\right)^{1/2^k} 
= \left(\avg{P_1,\ldots,P_k}{\avg{Q}{\oexp{P_1\cdots P_k f,
  Q}}}\right)^{1/2^k}\\
& = \left(\Pr_{P_1,\ldots,P_k}\left[ f \cdot \prod_{i} P_i \in \RM_q(n,d+ek)\right]\right)^{1/2^k}
\end{align*}
\fi
where the first inequality follows from Jensen's inequality and the second from the Squaring trick. For the third equality, we have used the fact that since $k < q$, the polynomials $k!P_1\cdots P_k$ and $P_1\cdots P_k$ are distributed identically.

The corollary now follows from \cref{thm:test-e-k}.
\end{proof}

\section*{Acknowledgements.} We thank Madhu Sudan for many encouraging
discussions and feedback. We also thank the anonymous reviewers of
FSTTCS 2016 for many corrections and pointing out a weakness in a
previous version of \cref{lem:robustdg}. Finally, we thank the anonymous reviewers for the IEEE Transactions on Information Theory for their insightful comments.

\ifIEEEtr
\else
{\small 
\bibliographystyle{prahladhurl}
\bibliography{robustDG-bib}

\newcommand{\etalchar}[1]{$^{#1}$}
\begin{thebibliography}{AKK{\etalchar{+}}05}

\bibitem[AKK{\etalchar{+}}05]{AlonKKLR2005}
\textsc{Noga Alon}, \textsc{Tali Kaufman}, \textsc{Michael Krivelevich},
  \textsc{Simon Litsyn}, and \textsc{Dana Ron}.
\newblock \href{http://dx.doi.org/10.1109/TIT.2005.856958} {\emph{Testing
  {R}eed-{M}uller codes}}.
\newblock {IEEE} Trans.\ Inform.\ Theory, 51(11):4032--4039, 2005.
\newblock (Preliminary version in {\em 7th RANDOM}, 2003).

\bibitem[ALM{\etalchar{+}}98]{AroraLMSS1998}
\textsc{Sanjeev Arora}, \textsc{Carsten Lund}, \textsc{Rajeev Motwani},
  \textsc{Madhu Sudan}, and \textsc{Mario Szegedy}.
\newblock \href{http://dx.doi.org/10.1145/278298.278306} {\emph{Proof
  verification and the hardness of approximation problems}}.
\newblock J. ACM, 45(3):501--555, May 1998.
\newblock (Preliminary version in {\em 33rd FOCS}, 1992).
\newblock
  \href{https://eccc.weizmann.ac.il/eccc-reports/1998/TR98-008}{\path{eccc:1998/TR98-008}}.

\bibitem[Aro94]{Arora1994}
\textsc{Sanjeev Arora}.
\newblock \href{https://www.cs.princeton.edu/~arora/pubs/thesis.pdf}
  {\emph{Probabilistic checking of proofs and the hardness of approximation
  problems}}.
\newblock Ph.D. thesis, University of California, Berkeley, 1994.

\bibitem[AS98]{AroraS1998}
\textsc{Sanjeev Arora} and \textsc{Shmuel Safra}.
\newblock \href{http://dx.doi.org/10.1145/273865.273901} {\emph{Probabilistic
  checking of proofs: A new characterization of~{NP}}}.
\newblock J. ACM, 45(1):70--122, January 1998.
\newblock (Preliminary version in {\em 33rd FOCS}, 1992).

\bibitem[AS03]{AroraS2003}
\textsc{Sanjeev Arora} and \textsc{Madhu Sudan}.
\newblock \href{http://dx.doi.org/10.1007/s00493-003-0025-0} {\emph{Improved
  low-degree testing and its applications}}.
\newblock Combinatorica, 23(3):365--426, 2003.
\newblock (Preliminary version in {\em 29th STOC}, 1997).
\newblock
  \href{https://eccc.weizmann.ac.il/eccc-reports/1997/TR97-003}{\path{eccc:1997/TR97-003}}.

\bibitem[BFLS91]{BabaiFLS1991}
\textsc{L{\'a}szl{\'o} Babai}, \textsc{Lance Fortnow}, \textsc{Leonid~A.
  Levin}, and \textsc{Mario Szegedy}.
\newblock \href{http://dx.doi.org/10.1145/103418.103428} {\emph{Checking
  computations in polylogarithmic time}}.
\newblock In \emph{Proc.\ $23$rd ACM Symp.\ on Theory of Computing (STOC)},
  pages 21--31. 1991.

\bibitem[BGH{\etalchar{+}}15]{BarakGHMRS2015}
\textsc{Boaz Barak}, \textsc{Parikshit Gopalan}, \textsc{Johan H{\aa}stad},
  \textsc{Raghu Meka}, \textsc{Prasad Raghavendra}, and \textsc{David Steurer}.
\newblock \href{http://dx.doi.org/10.1137/130929394} {\emph{Making the long
  code shorter}}.
\newblock SIAM J. Comput., 44(5):1287--1324, 2015.
\newblock (Preliminary version in {\em 53rd FOCS}, 2012).
\newblock \href{http://arxiv.org/abs/1111.0405}{\path{arXiv:1111.0405}},
  \href{https://eccc.weizmann.ac.il/eccc-reports/2011/TR11-142}{\path{eccc:2011/TR11-142}}.

\bibitem[BKS{\etalchar{+}}10]{BhattacharyyaKSSZ2010}
\textsc{Arnab Bhattacharyya}, \textsc{Swastik Kopparty}, \textsc{Grant
  Schoenebeck}, \textsc{Madhu Sudan}, and \textsc{David Zuckerman}.
\newblock \href{http://dx.doi.org/10.1109/FOCS.2010.54} {\emph{Optimal testing
  of {R}eed-{M}uller codes}}.
\newblock In \emph{Proc.\ $51$st IEEE Symp.\ on Foundations of Comp.\ Science
  (FOCS)}, pages 488--497. 2010.
\newblock \href{http://arxiv.org/abs/0910.0641}{\path{arXiv:0910.0641}},
  \href{https://eccc.weizmann.ac.il/eccc-reports/2009/TR09-086}{\path{eccc:2009/TR09-086}}.

\bibitem[CLO15]{CoxLittleOShea}
\textsc{David~A Cox}, \textsc{John Little}, and \textsc{Donal O'Shea}.
\newblock \href{http://dx.doi.org/10.1007/978-3-319-16721-3} {\emph{Ideals,
  Varieties, and Algorithms}}.
\newblock Springer, 3rd edition, 2015.

\bibitem[DG15]{DinurG2015}
\textsc{Irit Dinur} and \textsc{Venkatesan Guruswami}.
\newblock \href{http://dx.doi.org/10.1007/s11856-015-1231-3} {\emph{{PCP}s via
  the low-degree long code and hardness for constrained hypergraph coloring}}.
\newblock Israel J.\ Math., 209:611--649, 2015.
\newblock (Preliminary version in {\em 54th FOCS}, 2013).
\newblock
  \href{https://eccc.weizmann.ac.il/eccc-reports/2013/TR13-122}{\path{eccc:2013/TR13-122}}.

\bibitem[FGL{\etalchar{+}}96]{FeigeGLSS1996}
\textsc{Uriel Feige}, \textsc{Shafi Goldwasser}, \textsc{L{\'a}szl{\'o}
  Lov{\'a}sz}, \textsc{Shmuel Safra}, and \textsc{Mario Szegedy}.
\newblock \href{http://dx.doi.org/10.1145/226643.226652} {\emph{Interactive
  proofs and the hardness of approximating cliques}}.
\newblock J. ACM, 43(2):268--292, March 1996.
\newblock (Preliminary version in {\em 32nd FOCS}, 1991).

\bibitem[FS95]{FriedlS1995}
\textsc{Katalin Friedl} and \textsc{Madhu Sudan}.
\newblock \href{http://dx.doi.org/10.1109/ISTCS.1995.377032} {\emph{Some
  improvements to total degree tests}}.
\newblock In \emph{Proc.\ $3$rd Israel Symp.\ on Theoretical and Computing
  Systems}, pages 190--198. 1995.
\newblock (See arXiv for corrected version).
\newblock \href{http://arxiv.org/abs/1307.3975}{\path{arXiv:1307.3975}}.

\bibitem[GHH{\etalchar{+}}17]{GuruswamiHHSV2017}
\textsc{Venkat Guruswami}, \textsc{Prahladh Harsha}, \textsc{Johan H{\aa}stad},
  \textsc{Srikanth Srinivasan}, and \textsc{Girish Varma}.
\newblock \href{http://dx.doi.org/10.1137/140995520}
  {\emph{Super-polylogarithmic hypergraph coloring hardness via low-degree long
  codes}}.
\newblock SIAM J. Comput., 46(1):132--159, 2017.
\newblock (Preliminary version in {\em 46th STOC}, 2014).
\newblock \href{http://arxiv.org/abs/1311.7407}{\path{arXiv:1311.7407}}.

\bibitem[GS06]{GoldreichS2006}
\textsc{Oded Goldreich} and \textsc{Madhu Sudan}.
\newblock \href{http://dx.doi.org/10.1145/1162349.1162351} {\emph{Locally
  testable codes and {PCP}s of almost linear length}}.
\newblock J. ACM, 53(4):558--655, 2006.
\newblock (Preliminary version in {\em 43rd FOCS}, 2002).
\newblock
  \href{https://eccc.weizmann.ac.il/eccc-reports/2002/TR02-050}{\path{eccc:2002/TR02-050}}.

\bibitem[HS16]{HarshaS2016-rm}
\textsc{Prahladh Harsha} and \textsc{Srikanth Srinivasan}.
\newblock \href{http://dx.doi.org/10.4230/LIPIcs.FSTTCS.2016.17} {\emph{Robust
  multiplication-based tests for {R}eed-{M}uller codes}}.
\newblock In \textsc{Akash Lal}, \textsc{S.~Akshay}, \textsc{Saket Saurabh},
  and \textsc{Sandeep Sen}, eds., \emph{Proc. $36$th IARCS Annual Conf.\ on
  Foundations of Software Tech.\ and Theoretical Comp.\ Science (FSTTCS)},
  volume~65 of \emph{LIPIcs}, pages 17:1--17:14. Schloss Dagstuhl, 2016.
\newblock \href{http://arxiv.org/abs/1612.03086}{\path{arXiv:1612.03086}}.

\bibitem[HSS13]{HaramatySS2013}
\textsc{Elad Haramaty}, \textsc{Amir Shpilka}, and \textsc{Madhu Sudan}.
\newblock \href{http://dx.doi.org/10.1137/120879257} {\emph{Optimal testing of
  multivariate polynomials over small prime fields}}.
\newblock SIAM J. Comput., 42(2):536--562, 2013.
\newblock (Preliminary version in {\em 52nd FOCS}, 2011).
\newblock
  \href{https://eccc.weizmann.ac.il/eccc-reports/2011/TR11-059}{\path{eccc:2011/TR11-059}}.

\bibitem[Hua15]{Huang2015}
\textsc{Sangxia Huang}.
\newblock \emph{$2^{(\log N)^{1/10-o(1)}}$ hardness for hypergraph coloring},
  2015.
\newblock (manuscript).
\newblock \href{http://arxiv.org/abs/1504.03923}{\path{arXiv:1504.03923}}.

\bibitem[KLP68]{KasamiLP1968}
\textsc{Tadao Kasami}, \textsc{Shu Lin}, and \textsc{W.~Wesley Peterson}.
\newblock \href{http://dx.doi.org/10.1109/TIT.1968.1054226} {\emph{Polynomial
  codes}}.
\newblock {IEEE} Trans.\ Inform.\ Theory, 14(6):807--814, 1968.

\bibitem[KR06]{KaufmanR2006}
\textsc{Tali Kaufman} and \textsc{Dana Ron}.
\newblock \href{http://dx.doi.org/10.1137/S0097539704445615} {\emph{Testing
  polynomials over general fields}}.
\newblock SIAM J. Comput., 36(3):779--802, 2006.
\newblock (Preliminary version in {\em 45th FOCS}, 2004).

\bibitem[KS17]{KhotS2017}
\textsc{Subhash Khot} and \textsc{Rishi Saket}.
\newblock \href{http://dx.doi.org/10.1137/15100240X} {\emph{Hardness of
  coloring 2-colorable 12-uniform hypergraphs with $2^{(\log n)^{\Omega(1)}}$
  colors}}.
\newblock SIAM J. Comput., 46(1):235--271, 2017.
\newblock (Preliminary version in {\em 55th FOCS}, 2014).
\newblock
  \href{https://eccc.weizmann.ac.il/eccc-reports/2014/TR14-051}{\path{eccc:2014/TR14-051}}.

\bibitem[Lin99]{vanLint}
\textsc{Jacobus~Hendricus van Lint}.
\newblock \href{http://dx.doi.org/10.1007/978-3-642-58575-3}
  {\emph{Introduction to Coding Theory}}.
\newblock Springer, 3rd edition, 1999.

\bibitem[RS96]{RubinfeldS1996}
\textsc{Ronitt Rubinfeld} and \textsc{Madhu Sudan}.
\newblock \href{http://dx.doi.org/10.1137/S0097539793255151} {\emph{Robust
  characterizations of polynomials with applications to program testing}}.
\newblock SIAM J. Comput., 25(2):252--271, April 1996.
\newblock (Preliminary version in {\em 23rd STOC}, 1991 and {\em 3rd SODA},
  1992).

\bibitem[RS97]{RazS1997}
\textsc{Ran Raz} and \textsc{Shmuel Safra}.
\newblock \href{http://dx.doi.org/10.1145/258533.258641} {\emph{A sub-constant
  error-probability low-degree test, and a sub-constant error-probability {PCP}
  characterization of~{NP}}}.
\newblock In \emph{Proc.\ $29$th ACM Symp.\ on Theory of Computing (STOC)},
  pages 475--484. 1997.

\bibitem[Var15]{Varma2015}
\textsc{Girish Varma}.
\newblock \href{http://dx.doi.org/10.4086/cjtcs.2015.003} {\emph{Reducing
  uniformity in {K}hot-{S}aket hypergraph coloring hardness reductions}}.
\newblock Chic.\ J.\ Theoret.\ Comput.\ Sci., 2015(3):1--8, 2015.
\newblock \href{http://arxiv.org/abs/1408.0262}{\path{arXiv:1408.0262}}.

\end{thebibliography}
}
\fi

\end{document}